\documentclass[12pt]{article}
\usepackage{amssymb}
\usepackage{amsmath}
\usepackage{amsfonts}
\usepackage{amsthm}
\usepackage{mathrsfs}
\usepackage{color}
\usepackage[displaymath,mathlines,pagewise]{lineno}

\newcommand{\bea}{\begin{eqnarray}}
\newcommand{\eea}{\end{eqnarray}}

\newcommand{\be}{\begin{equation}}
\newcommand{\ee}{\end{equation}}

\theoremstyle{plain}
\newtheorem{theorem}{Theorem}
\newtheorem{corollary}[theorem]{Corollary}
\newtheorem{proposition}[theorem]{Proposition}
\newtheorem{lemma}[theorem]{Lemma}

\theoremstyle{definition}
\newtheorem{definition}[theorem]{Definition}

\newcommand{\beas}{\begin{eqnarray*}}
\newcommand{\eeas}{\end{eqnarray*}}

\numberwithin{equation}{section}
\numberwithin{theorem}{section}
\begin{document}

\centerline{\Large {\bf A time-dependent energy-momentum method}}
\vskip 0.5cm

\centerline{ J. de Lucas  and B.M. Zawora}
\vskip 0.5cm

\centerline{ Department of Mathematical Methods in Physics, University of Warsaw,}
\medskip
\centerline{ul. Pasteura 5, 02-093 Warszawa, Poland}
\medskip

\vskip 1cm

\begin{abstract}
    We devise a generalisation of the energy momentum-method for studying the stability of non-autonomous Hamiltonian systems with a Lie group of Hamiltonian symmetries. A generalisation of the relative equilibrium point notion to a non-autonomous realm is provided and studied. Relative equilibrium points of a class of non-autonomous Hamiltonian systems are described via foliated Lie systems, which opens a new field of application of such systems of differential equations. We reduce non-autonomous Hamiltonian systems via the Marsden--Weinstein theorem and we provide conditions ensuring the stability of the projection of relative equilibrium points to the reduced space. As a byproduct, a geometrical extension of notions and results from Lyapunov stability theory on linear spaces to  manifolds is provided. As an application, we study a class of mechanical systems, the hereafter called almost-rigid bodies, which covers rigid bodies as a particular instance. 
\end{abstract}




\section{Introduction}

Symplectic geometry has a fruitful history of applications to classical mechanics \cite{AM78,Go80,GS90}. Its origin can be traced back to the pioneering works by Lagrange, who carefully analysed  the rotational motion of mechanical systems  \cite{La09}. 

Toward the end of the XXth century, the Marsden--Weinstein reduction theorem \cite{MW74} was devised so as to describe the reduction of Hamiltonian systems on a symplectic manifold admitting a certain Lie group of symmetries of the Hamiltonian of the system and the symplectic form of the manifold. This theorem, an improvement of previous ideas by Lie, Smale, and Cartan \cite{MW01}, led to relevant applications in classical mechanics as well as many extensions to other types of geometric structures \cite{Al89,CFZ07,MW86}.

Let $\Phi:G\times P\rightarrow P$ be a Lie group action having a family of Hamiltonian fundamental vector fields relative to a symplectic form $\omega$ on $P$, i.e. a {\it Hamiltonian Lie group action}, and leaving invariant $h\in C^\infty(P)$. Weinstein and Marsden used $\Phi$ and $\omega$ to define the so-called {\it momentum map} ${\bf J}:P\rightarrow \mathfrak{g}^*$, where $\mathfrak{g}^*$ is the dual to the Lie algebra, $\mathfrak{g}$, of $G$. By assuming ${\bf J}$ to be {\it equivariant} \cite[p. 279]{AM78} relative to $\Phi$ and the coadjoint action, Marsden and Weinstein reduced the Hamiltonian problem $h$ on $P$ to a problem in the space of orbits $P_\mu:={\bf J}^{-1}(\mu)/G_\mu$, for a regular point $\mu\in \mathfrak{g}^*$ of ${\bf J}$, relative to the isotropy subgroup $G_\mu\subset G$ of $\mu$ acting freely and properly on ${\bf J}^{-1}(\mu)$. Remarkably, $P_\mu$ admits a canonically defined symplectic form, $\omega_\mu$, while the Hamiltonian system $h$ on $P$ leads to a new one  on $P_\mu$ given by the unique function $k_\mu$ such that $k_\mu\circ \pi_\mu:=h$ on ${\bf J}^{-1}(\mu)$, where $\pi_\mu:{\bf J
}^{-1}(\mu)\rightarrow P_\mu$ is the quotient map.

The Hamiltonian system $k_\mu$ on $P_\mu$ has {\it equilibrium points}, i.e. stable points relative to the evolution given by the Hamilton equations for $k_\mu$ in $P_\mu$, that are the projection of not necessarily equilibrium points of $h$ on $P$, the referred to as {\it relative equilibrium points} of $h$ relative to $\Phi$ \cite{AM78,MS88}. It is interesting to study the properties of the solutions to the Hamilton equations of $h$ that project onto equilibrium points of $k_\mu$. It is also relevant to study the stability of the Hamilton equations for $k_\mu$ close to its equilibrium points. The energy-momentum method was developed to study these problems, which are autonomous  \cite{MS88}. Instead of analysing straightforwardly the reduced system on $P_\mu$, the energy-momentum method studies the Hamiltonian problem on $P_\mu$ via the properties of the initial function $h$ on $P$, which is easier as it avoids, among other difficulties, the necessity of constructing $P_\mu$ and $k_\mu$ explicitly (cf. \cite{MS88}). 

There have been several generalisations of the energy-momentum method as well as some improvements and many applications of the developed theories (see \cite{LMRS90,OPR05,SLM91,ST92,TWP19} and references therein). In this work, we present a time-dependent generalisation of the energy-momentum method on symplectic manifolds. The Marsden--Weinstein theorem can also be applied to a time-dependent function $h:\mathbb{R}\times P\rightarrow \mathbb{R}$ that is invariant relative to a Hamiltonian Lie group action $\Phi$ with respect to a  symplectic form $\omega$ on $P$ (cf. \cite{MW74}). We here suggest a definition of a relative equilibrium point for $h$ relative to $\Phi$. We also study the structure of the space of relative equilibrium points in $P$. 

Our work proves that the dynamics of $h$ on its space of relative equilibrium points can be described, in certain cases, through {\it foliated Lie systems} \cite{CGM00}. The work \cite{CGM00} details  the potential application of   foliated Lie systems in integrable Hamiltonian systems and other rather theoretical  examples. Our work shows another potential field of application of foliated Lie systems.

The stability of the Hamilton equations for $k_\mu$, obtained through the reduction of $h:\mathbb{R}\times P\rightarrow \mathbb{R}$ via the Marsden--Weinstein theorem, close to its equilibrium points is addressed by studying the properties of $h$ and some of its restrictions and reductions to ${\bf J}^{-1}(\mu)$ or $P_\mu$. Our theory retrieves quite easily the results of the classical energy-momentum method, which deals with autonomous Hamiltonian systems. Our time-dependent energy-momentum method requires the use of time-dependent Lyapunov stability theory \cite{Kh87,Vi02}, which is much more involved than standard techniques employed in the energy-momentum method. To illustrate this fact, one can compare Lemma \ref{Lemm:MegaLemma}, Theorems \ref{Th:StabilityCon} and \ref{Th:StabilityCon2} with the standard results in \cite{MS88}. As a byproduct, our work also extends some results of the Lyapunov stability theory on $\mathbb{R}^n$ to manifolds. It is worth stressing that the extension of Lyapunov theory to manifolds has drawn attention just during the last fifteen years (cf. \cite{Mo11} and references therein).

As a potential application, we study an orbiting mechanical system that, as a particular case, retrieves the rigid body and the standard theory that can be found, for instance, in the classical work \cite{MS88}. Due to the many applications of the energy-momentum method and their generalisations \cite{TWP19}, our results may have numerous potential applications. 

The work goes as follows. Section 2 details a  generalisation of some fundamental results on Lyapunov stability on $\mathbb{R}^n$ to manifolds. Section 3 describes some basic notions on  symplectic manifolds and the conventions to be used hereafter. Section 3 also gives some generalisations to the $t$-dependent realm of results on autonomous Hamiltonian systems. Section 4 generalises the notion of relative equilibrium point to time-dependent Hamiltonian systems. Section 5 studies the relation between the manifold of relative equilibrium points and foliated Lie systems. Section 6 analyses the stability of trajectories around equilibrium points of non-autonomous Hamiltonian systems. Section 7 links the properties of stable points in $P_\mu$ and their associated relative equilibrium points in ${\bf J}^{-1}(\mu)$. Section 8 details an example of our theory. Finally, our results are summarised and an outlook of further research is presented in Section 9.  

\section{Fundamentals on the Lyapunov stability of non-autonomous systems}

From now on, and if not otherwise stated, we assume all structures to be smooth, real, and globally defined. This stresses the key ideas of our presentation. Additionally, manifolds are assumed to be finite-dimensional and connected. 

Let us provide an adaptation of some basic results on the Lyapunov stability theory on $\mathbb{R}^n$  \cite{Ha67,Kh87,MLS94,Vi02} to manifolds. This will allow us to use this theory to study differential equations on manifolds. It is worth stressing that, as far as we know, the idea of extending Lyapunov theory to manifolds is very recent and only a couple of works on the topic  have been published so far (see \cite{Mo11} and references therein). It will be simple to see that our approach retrieves the standard Lyapunov theory when restricted to problems on a Euclidean space $\mathbb{R}^n$. Our final aim is to apply these techniques to studying the stability of the Hamilton equations of reduced $t$-dependent Hamiltonian systems by the Marsden--Weinstein theorem \cite{MW74} close to its equilibrium points. 

To generalise Lyapunov theory on linear spaces to manifolds, one has to find a substitute for the norm on linear spaces, which is extensively used in Lyapunov theory on linear spaces. This norm comes from a Euclidean metric on linear spaces. Let us show how to extend this structure to manifolds. Euclidean metrics give rise to Riemannian metrics. In fact, any manifold $P$ admits a Riemannian metric due to the existence of partitions of unity \cite{BC64}. Let us assume $P$ to be endowed with a Riemannian metric $g$. Then, the  distance between two points $x_1,x_2\in P$, let us say $d(x_1,x_2)$, is given by 
\[d(x_1,x_2):=\!\!\!\!\!\!\inf_{\substack{\tiny \gamma:[0,1]\rightarrow P\\\gamma(0)=x_1,\gamma(1)=x_2}}\!\!\!\!\!\!{\rm length}(\gamma), 
\]
where ${\rm length}(\gamma)$ is the length of a curve $\gamma$ in $P$ relative to $g$ (see \cite{Le09}). Let $B_{r,x_e}$ be the ball of radius $r$ around $x_e\in P$ relative to the distance induced by $g$, namely $B_{r,x_e}:=\{x\in P\,:\,d(x,x_e)<r\}$ with $r>0$. It can be proved that the topology induced by a Riemannian metric on $P$ is the same as the topology of the manifold $P$ \cite{KN96}. Then, given a point $x\in P$, every map on $P$ containing $x$ gives a homomorphism (in a topological sense) to an open subset of $\mathbb{R}^n$. Hence, on an open coordinate neighbourhood of $x\in P$, the topology of the manifold is equivalent to the topology of an open subset in $\mathbb{R}^n$ given by the standard norm in $\mathbb{R}^n$. Consequently, topological properties on an open coordinate neighbourhood of each $x\in P$ can be studied by using the norm on $\mathbb{R}^n$.

Hereafter, $t$ stands for the physical time.
Let $X:(t,x)\in \mathbb{R}\times P\mapsto X(t,x)\in TP$ be a $t$-dependent vector field on $P$, namely a $t$-parametric family of vector fields $X_t:x\in P\mapsto X(t,x)\in TP$ on $P$ with $t\in \mathbb{R}$ (see \cite{LS20} for details). 
Let us consider the following non-autonomous dynamical system
\begin{equation} 
\label{eq:1}
\frac{dx}{dt}=X(t,x), \qquad \forall x\in P,\qquad \forall t\in \mathbb{R},
\end{equation}
where $X$ is assumed to be smooth and then  (\ref{eq:1}) satisfies the conditions of the Theorem of existence and uniqueness of solutions \cite[Theorem 2.1.2]{AM78}.

Let $\bar{\mathbb{R}}:=\mathbb{R}_+\cup\{0\}$ be the space of non-negative real numbers. We hereafter write $I_{t'}:=[t',\infty[$ for any $t'\in \mathbb{R}$ and $I_{-\infty}:=\mathbb{R}$. A point $x_e\in P$ is an {\it equilibrium point} of (\ref{eq:1}) if $X(t,x_e)=0$ for every $t\in \mathbb{R}$. An equilibrium point $x_e$ is {\it stable} from $t^0\in \mathbb{R}$ if, for every $t_0\in I_{t^0}$ and any ball $B_{\epsilon,x_e}$, there exists a ball of radius $\delta(t_0,\epsilon)$, namely $B_{\delta(t_0,\epsilon),x_e}$, such that every solution $x(t)$ to (\ref{eq:1}) with $x(t_0)\in B_{\delta(t_0,\epsilon),x_e}$ satisfies that $x(t)\in B_{\epsilon,x_e}$ for all time $t\in I_{t_0}$. If $t^0$ is not hereafter explicitly detailed, we assume that $t^0=-\infty$. An equilibrium point $x_e\in P$ is {\it uniformly stable} from $t^0\in \mathbb{R}$ if for every $\epsilon>0$, one can choose $\delta(t_0,\epsilon)$, with $t_0\in I_{t^0}$, to be independent of $t_0$. An equilibrium point is {\it unstable} from $t^0$ if it is not stable from $t^0$.

An equilibrium point $x_e$ is {\it asymptotically stable} from $t^0$ if $x_e$ is stable and for every $t_0\in I_{t^0}$ there exists an open neighbourhood $B_{r(t_0),x_e}$ of $x_e$ such that every solution $x(t)$ to (\ref{eq:1}) with $x(t_0)\in B_{r(t_0),x_e}$ converges to $x_e$. Moreover, $x_e$ is {\it uniformly asymptotically stable} from $t^0$ if it is asymptotically stable and $r(t_0)$ can be chosen to be independent of $t_0\in I_{t^0}$ and the convergence to $x_e$ is uniform relative to $x$ in $B_{r,x_e}$ and $t\in I_{t^0}$ (for more details, see \cite[p. 140]{Vi02}).

\begin{definition}
\label{lpdf}
A continuous function $M:I_{t^0}\times P\rightarrow\mathbb{R}$ is a {\it locally positive definite function} ({\it lpdf}) {\it at an equilibrium point $x_e$ from $t^0\in \mathbb{R}$} if, for some $r>0$ and some continuous, strictly increasing function $\alpha:\bar{\mathbb{R}}\rightarrow\mathbb{R}$ with $\alpha(0)=0$, one has that 
\[
    M(t,x_e)=0,\quad M(t,x)\geq \alpha(d(x,x_e)),\quad\forall t\in I_{t^0},\quad \forall x\in B_{r,x_e}.
\]
\end{definition}

\begin{definition}
\label{DecrescentFunction}
A continuous function $M:I_{t^0}\times P\rightarrow\mathbb{R}$ is {\it  decrescent at an equilibrium point $x_e$} from $t^0\in \mathbb{R}$ if, for some $s >0$ and some continuous, strictly increasing function $\beta:\bar{\mathbb{R}}\rightarrow\mathbb{R}$ with $\beta(0)=0$, is fulfilled that 
\[
    M(t,x)\leq \beta(d(x,x_e)),\quad\forall t\in I_{t^0},\quad \forall x\in B_{s,x_e}.
\]
\end{definition}
Although  Definitions \ref{lpdf} and \ref{DecrescentFunction} concern a continuous function $M$, and so is in the literature \cite{Ha67,Kh87,MLS94,Vi02}, it is more relevant for our purposes to consider that  $M(t,x)$ is a $C^1$ function. That is why we hereafter assume that $M$ is $C^1$.
We define $\dot M:I_{t^0}\times P\rightarrow \mathbb{R}$ to be a function so that  $\dot{M}(\hat t,\hat x)$, for $(\hat t,\hat x)\in I_{t^0}\times P$, is the time derivative of $M(t,x(t))$ at $t=\hat t$ along the particular solution $x(t)$  of \eqref{eq:1} with initial condition $x(\hat t)=\hat x$, i.e. 
\begin{equation}\label{Eq:Mdot}
\dot{M}(\hat t,\hat x):=\frac{d}{dt}\bigg|_{t=\hat t}M(t,x(t))=\frac{\partial M}{\partial t}(\hat t,\hat x)+\sum_{i=1}^{\dim P}\frac{\partial M}{\partial x^i}(\hat t,\hat x) X^i(\hat t,\hat x),
\end{equation}
where $\{x^1,\ldots,x^{\dim P}\}$ is  a local coordinate system in $P$ around $\hat x$ and $X^1,\ldots,X^{\dim P}$ are the coordinates of $X$ in the basis of vector fields associated with the given local  coordinates.

Above definitions are significant to understand Theorem \ref{Th:BasicTheoremOfLyapunov}, which allows us to determine the stability of (\ref{eq:1}) by studying the properties of an appropriate associated function.

For the sake of completeness and clarity, we shall write down an extension to manifolds of some classical results for linear spaces \cite{Vi02} given by the following theorems.

\begin{theorem}
\label{Mstable}
An equilibrium point $x_e\in P$ of the system (\ref{eq:1}) is stable from $t^0$ if there exists a  lpdf $C^1$-function $M:I_{t^0}\times P\rightarrow \mathbb{R}$ from $t^0\in \mathbb{R}$ and a constant $r>0$ such that
\[
\dot{M}(t,x)\leq 0,\qquad \forall t\in I_{t^0},\quad \forall x\in B_{r,x_e}.
\]
\end{theorem}

\begin{proof}
Since the function $M$ is lpdf from $t^0$ by assumption, Definition \ref{lpdf} yields that there exists a continuous strictly increasing function from $t^0$, let us say $\alpha:\bar{\mathbb{R}}\rightarrow \mathbb{R}$, and a constant $
s>0$ such that
\[
\alpha(d(x,x_e))\leq M(t,x), \qquad \forall t\in I_{t^0},\quad\forall x\in B_{s,x_e}.
\]
Let us show that $x_e$ is stable from $t^0$, i.e. there exists, for any $\epsilon>0$, $t_0\in I_{t^0}$, and $t\in I_{t_0}$,  a $\delta(t_0,\epsilon)=:\delta$ such that if $x(t)$ is the particular solution of the system (\ref{eq:1}) with initial condition $x_0:=x(t_0)$, then
\[
d(x_0,x_e)<\delta \implies d(x(t),x_e)< \epsilon,\qquad \forall t\in  I_{t_0}.
\]
Let us choose $\epsilon$, $t_0$, and let $\mu:=\min(\epsilon, r, s)$. Then, there exists $\delta>0$ so that
\[
\sup_{d(x,x_e)<\delta}M(t_0,x)<\alpha(\mu).
\]
This is possible since $\alpha(\mu)>0$ and $\lim_{\delta\rightarrow 0^+}\sup_{d(x,x_e)<\delta}M(t_0,x)=0$. To show that $\delta$ guarantees the stability of $x_e$, suppose $d(x_0,x_e)<\delta$. Then, $M(t_0,x_0)\leq \sup_{d(x,x_e)<\delta}M(t_0,x)<\alpha(\mu)$.

Let us assume for the time being that $x(t)$ belongs to $B_{\mu,x_e}$ for every $t\in  I_{t_0}$. Then, $B_{\mu,x_e}\subset B_{r,x_e}$ and $\dot{M}(t,x(t))\leq 0$ and from the assumption that $M(t,x)$ is a $C^1$-function, it follows that $M(t,x(t))-M(t_0,x_0)\leq0$. Thus,
\begin{equation}
\label{M(t,x)<M(t0,x0)}
M(t,x(t))\leq M(t_0,x_0)< \alpha(\mu),\qquad \forall t\in I_{t_0}.  
\end{equation}
Since $x(t)\in B_{\mu,x_e}\subset B_{s,x_e}$ for $t\in I_{t_0}$ by assumption, we also have that
\[
\alpha(d(x(t),x_e))\leq M(t,x(t)),\qquad \forall t\in I_{t_0}. 
\]
Hence, from the last two inequalities, one obtains
\[
\alpha(d(x(t),x_e))< \alpha(\mu),\qquad \forall t\in I_{t_0}. 
\]
Since $\alpha$ is a strictly increasing function, it follows that
\begin{equation}
\label{x<e}
d(x(t),x_e)< \mu \leq \epsilon,\qquad \forall t\in I_{t_0}. 
\end{equation}
Hence, $x_e$ is a stable equilibrium under the assumption of $x(t)$ belonging to $B_{\mu,x_e}$ for every $t\in I_{t_0}$. Let us prove that this assumption always holds indeed. 

Assume that $T:=\min\{ t\in \mathbb{R}\,:\,d(x(t),x_e)\geq \mu\}$ (it is well defined, since $x(t)$ is continuous). By definition of $T$, it turns out that
\[
d(x(t),x_e)< \mu,\qquad \forall t \in [t_0,T),
\]
and, by continuity, $d(x(T),x_e)=\mu$. Since $\mu\leq r$, it follows that 
\[
\dot{M}(t,x(t))\leq 0,\qquad \forall t\in [t_0,T).
\]
Hence, from  the fact that $M$ is a $C^1$-function, one obtains
\begin{equation}
\label{M<g}
M(T,x(T))\leq M(t_0,x_0)<\alpha(\mu).
\end{equation}
However, $\mu\leq s$ and
\begin{equation}
\label{M=g}
M(T,x(T))\geq \alpha(d(x(T),x_e))=\alpha(\mu).
\end{equation}
Equations (\ref{M<g}) and (\ref{M=g}) are in contradiction, which gives that no such $T$ exists. Thus, (\ref{x<e}) is true.
\end{proof}

\begin{theorem}
\label{Muniformlystable}
An equilibrium point $x_e$ of system (\ref{eq:1}) is uniformly stable from $t^0$ if there exists a $C^1$, lpdf and also decrescent function $M:I_{t^0}\times P\rightarrow \mathbb{R}$ from $t^0$ and a constant $r>0$ such that
\[
\dot{M}(t,x)\leq 0,\qquad \forall t\in I_{t^0},\quad \forall x\in B_{r,x_e}.
\]
\end{theorem}

\begin{proof}
The proof of this theorem will be only sketched, because is very similar to the proof of Theorem \ref{Mstable}. 
Since  $M$ is decrescent from $t^0$ by assumption, Definition \ref{DecrescentFunction} yields that there exists a continuous, strictly increasing function $\beta:\bar{\mathbb{R}}\rightarrow \mathbb{R}$ with $\beta(0)=0$ and a constant $s>0$ such that
\[
M(t,x)\leq \beta(d(x,x_e)),\qquad \forall t\in I_{t^0},\quad\forall x\in B_{s,x_e}.
\] 
Then, we define 
\[
 \omega(\delta):=\sup_{d(x,x_e)<\delta,\,\, t\in I_{t^0}} M(t,x).
\]
Such a function is well defined for $\delta<s$ because $M(t,x)$ is decrescent and $\omega(\delta)\leq \beta(\delta)$. Moreover, $\omega(\delta)$ is non-decreasing and 
\[
\lim_{\delta\rightarrow 0^+}\omega(\delta)=\lim_{\delta\rightarrow 0^+}\sup_{d(x,x_e)<\delta,\,\, t\in I_{t^0}}M(t,x)\leq \lim_{\delta\rightarrow 0^+}\beta(\delta)=0.
\]
Since $M$ is a ldpf function, consider the function
$\alpha:\bar{\mathbb{R}}\rightarrow \mathbb{R}$ and the constant $s_1>0$ such that
\[
\alpha(d(x,x_e))\leq M(t,x),\qquad \forall t\in I_{t^0},\quad \forall x\in B_{s_1,x_e}.
\]
Set some $\epsilon>0$. Define $\mu:=\min(\epsilon,r,s,s_1)$. Let us choose $\delta$ such that $
\beta(\delta)<\alpha(\mu).$ 
The rest of the proof is analogous to the previous theorem, including the proof that $x(t)$ stays in $B_{\mu,x_e}$ for all $t\geq t_0\geq t^0$ if $x(t_0)$ is contained in $B_{\mu,x_e}$.
\end{proof}

\begin{theorem}
The equilibrium point $x_e$ of system (\ref{eq:1}) is uniformly asymptotically stable from $t^0$ if there exists a decrescent, lpdf,  $C^1$-function $M:I_{t^0}\times P\rightarrow \mathbb{R}$ from $t^0$ such that $-\dot{M}$ is a lpdf from $t^0$.
\end{theorem}

\begin{proof}
Let $x(t)$ stands for a solution of system (\ref{eq:1}) with initial condition $x(t_0)=x_0$ for some $t_0\geq t^0$. Since $-\dot{M}$ is a lpdf function, by Definition \ref{lpdf} and the assumptions of our present theorem, there exists a continuous, strictly increasing function $\gamma:\bar{\mathbb{R}}\rightarrow\mathbb{R}$, with $\gamma(0)=0$, and a constant $s>0$ such that
\[
\dot{M}(t,x)\leq -\gamma(d(x,x_e)),\qquad \forall t\in I_{t^0},\quad \forall x\in B_{s,x_e}.
\]
Since $\gamma$ is a non-negative function,
\begin{equation}
\label{nonpositivitydotM}
\dot{M}(t,x)\leq 0,\qquad \forall t\in I_{t^0},\quad \forall x\in B_{s,x_e}.
\end{equation}
Thus, $\dot{M}$ satisfies the hypothesis of Theorem \ref{Muniformlystable} and $x_e$ becomes a uniformly stable equilibrium from $t^0$. Then, what is left to prove is that for every $\epsilon>0$ and $t_0\geq t^0$ there exists $T:=T(\epsilon)$ and $\delta>0$ such that every $x(t)$ with $x(t_0)\in B_{\delta,x_e}$ satisfies that $d(x(t),x_e)<\epsilon$ for all $t\geq T+t_0$. It is sufficient to show that such a constant $\delta$ exists. The latter condition can be rewritten as follows
\begin{equation}
\label{uniformlyattractive}
\forall \epsilon>0,\quad \exists\, \delta>0,\quad \exists T>0, \quad d(x_0,x_e)<\delta \implies d(x(t),x_e)<\epsilon,\quad\forall t\geq T+t_0.
\end{equation}
The assumptions of the present theorem yield that there are functions $\alpha, \beta:\bar{\mathbb{R}}\rightarrow\mathbb{R}$ and constants $k,l>0$ such that
\begin{equation}
\label{M>alpha}
\alpha(d(x,x_e))\leq M(t,x),\qquad\forall t\in I_{t^0},\quad\forall x\in B_{k,x_e},
\end{equation} 
\begin{equation}
\label{beta>m}
M(t,x)\leq\beta(d(x,x_e)),\qquad \forall t\in I_{t^0},\quad\forall x\in B_{l,x_e}. 
\end{equation}

Let us choose $r:=\min\{k,l,s,\epsilon\}$.  Let us define positive constants $\kappa_1,\kappa_2,T$ such that
\[
        \kappa_1<\beta^{-1}(\alpha(r)),\qquad
        \kappa_2<\min\{\beta^{-1}(\alpha(\epsilon)),\kappa_1\},\qquad 
        T:=\frac{\beta(\kappa_1)}{\gamma(\kappa_2)}.
\]

Let us prove that we can set $\delta=\kappa_2$ and $T$ satisfy (\ref{uniformlyattractive}).  Recall that every particular solution $x(t)$ to (\ref{eq:1}) with $x(t_0)=:x_0\in B_{\kappa_2,x_e}$ remains inside the ball $B_{r,x_e}$ for all $t\in I_{t_0}$ and $\kappa_2$ small enough. Indeed, the reasoning of the proof is as in the previous theorems. We can assume indeed that (\ref{M>alpha}), (\ref{beta>m})  apply to $B_{\kappa_2,x_e}$.

First, let us prove that
\begin{equation}
\label{[t0,t0+t]}
    d(x_0,x_e)<\kappa_1\implies d(x(t_1),x_e)<\kappa_2,\qquad \exists\, t_1\in [t_0,t_0+T].
\end{equation}
The proof proceeds by contradiction, namely suppose that
\begin{equation}
\label{a.a.1}
    d(x_0,x_e)<\kappa_1\qquad\wedge\qquad d(x(t),x_e)\geq\kappa_2,\qquad\forall t\in [t_0,t_0+T].
\end{equation}
Using (\ref{M>alpha}), (\ref{beta>m}), and (\ref{nonpositivitydotM}) in (\ref{a.a.1}), we can obtain the following inequalities
\[
    \beta(d(x_0,x_e))<\beta(\kappa_1),\qquad \gamma(d(x(t),x_e))\geq\gamma(\kappa_2),\qquad \alpha(\kappa_2)\leq \alpha(d(x(t),x_e)), 
\]
for all $t_0<t<t_0+T$ and $x_0\in B_{\kappa_2,x_e}$. Then,
\begin{multline*}
0<\alpha(\kappa_2)\leq M(t_0+T,x(t_0+T))=M(t_0,x_0)+\int^{t_0+T}_{t_0}\dot{M}(\tau,x(\tau))d\tau\leq\\\beta(d(x_0,x_e))-\int^{t_0+T}_{t_0}\gamma(d(x(\tau),x_e))d\tau\leq \beta(\kappa_1) - T\gamma(\kappa_2)=0.
\end{multline*}
This contradiction shows that (\ref{[t0,t0+t]}) is true. To complete the proof, suppose $t>t_0+T$. Inequality (\ref{M>alpha}) holds for all $t\in I_{t_0}$ and using (\ref{[t0,t0+t]}) one can choose such $t_1\in [t_0,t_0+T]$ that $\beta(d(x(t_1),x_e))<\beta(\kappa_2)$ is satisfied. Then, using (\ref{nonpositivitydotM}), we obtain
\[
\alpha(d(x(t),x_e))\leq M(t,x(t))\leq M(t_1,x(t_1))
\]
and
\[
M(t_1,x(t_1))\leq  \beta(d(x(t_1),x_e))<\beta(\kappa_2),
\]
and finally one can combine the last two inequalities to get
\[
\alpha(d(x(t),x_e))< \beta(\kappa_2)\leq \alpha(\epsilon),
\]
which establish (\ref{uniformlyattractive}) for $\delta=\kappa_2$ and ends the proof.
\end{proof}

The following theorem summarises the last three theorems in one theorem called the basic manifold Lyapunov's theorem.

\begin{theorem}
\label{Th:BasicTheoremOfLyapunov} {\bf (The basic manifold Lyapunov's theorem \cite{Kh87,MLS94,Vi02})}
Let $M:I_{t^0}\times P\rightarrow \mathbb{R}$ be a non-negative function, let $x_e\in P$ be an equilibrium point of (\ref{eq:1}), and let $\dot M$ stand for the function (\ref{Eq:Mdot}). Then, one has the following results:
\begin{enumerate}
    \item If $M$ is $C^1$ and lpdf from $t^0$ and $\dot{M}(t,x)\leq 0$ for $x$ locally around $x_e$ and for all $t\in I_{t^0}$, then $x_e$ is  stable.
    
    \item If $M$ is $C^1$, lpdf and decrescent from $t^0$, and   $\dot{M}(t,x)\leq 0$ locally around $x_e$ and for all $t\in I_{t^0}$, then $x_e$ is uniformly  stable.
    
    \item If $M$ is $C^1$, lpdf and decrescent from $t^0$, and $-\dot{M}(t,x)$ is locally positive definite around $x_e$ and $t\in I_{t^0}$, then $x_e$ is uniformly asymptotically stable.
    
\end{enumerate}
\end{theorem}

\section{Basics on symplectic geometry}
Let us review some known facts on symplectic geometry. At the same time, we are to establish the notions and sign conventions to be used hereafter while proving some non-autonomous extensions of classical results concerning autonomous Hamiltonian systems. For details on the topics and standard results provided in this section, we refer to \cite{AM78,Ca06,Va94}. 

A {\it  symplectic manifold} is a pair $(P,\omega)$, where $P$ is a manifold and $\omega$ is a closed differential two-form on $P$ that is {\it  non-degenerate}, namely the mapping $\widehat{\omega}:TP\mapsto T^*P$ of the form $\widehat \omega(v_p):=\omega_p(v_p,\cdot)\in T^*_pP$ for every $p\in P$ and every $v_p\in T_pP$, is a diffeomorphism. We  call $\omega$ a {\it  symplectic form}.

From now on, $(P,\omega)$ stands for a symplectic manifold.
The {\it  symplectic orthogonal} of a subspace $V_p\subset T_pP$ relative to $(P,\omega)$ is defined as
$V^{\perp_\omega}_p:=\{w_p\in T_pP\,:\,\omega_p(w_p,v_p)=0,\,\forall v_p\in V_p\}$. Let us now describe a specially important case of symplectic manifold. Let $Q$ be any manifold, let $\tau:T^*Q\rightarrow Q$ be the canonical projection, and let $\langle\cdot,\cdot\rangle$ be the pairing between covectors and tangent vectors on a manifold. The {\it canonical one-form} on $T^*Q$ is defined to be the differential one-form, $\theta_Q$, on $T^*Q$ given by
\[
(\theta_Q)_{\alpha_q}(v_{\alpha_q})\!:=\!\langle \alpha_q,T_{\alpha_q}\tau(v_{\alpha_q})\rangle,\quad \forall q\in Q,\quad \forall \alpha_q\in T^*_q{Q},\quad\forall v_{\alpha_q}\in T_{\alpha_q}(T^*Q).
\]
On local adapted coordinates $\{q^i,p_i\}_{i=1,\ldots,n}$ on $T^*Q$, one has  $\theta_Q:=\sum_{i=1}^n p_idq^i$. Then, $\omega_Q:=-d\theta_Q=\sum^n_{i=1}dq^i\land dp_i$ is a symplectic form, the referred to as {\it canonical symplectic form} on $T^*Q$. The symplectic manifold $(T^*Q,\omega_Q)$ is relevant in physical applications.

Let $\mathfrak{X}(P)$ be the Lie algebra of vector fields on $P$. A vector field $X\in \mathfrak{X}(P)$ is {\it  Hamiltonian} if the contraction of $\omega$ with $X$ is an exact differential one-form, i.e. $\iota_{X}\omega=df$ for some $f\in C^{\infty}(P)$.
Then, $f$ is called a {\it  Hamiltonian function} of $X$. Since $\omega$ is non-degenerate, every   $f\in C^\infty(P)$ is the Hamiltonian function of a unique Hamiltonian vector field $X_f$.
Then, the {\it  Cartan's magic formula} \cite[p. 115]{AM78} yields $\mathcal{L}_{X_f}\omega=\iota_{X_f}d\omega+d\iota_{X_f}\omega=0$, where $\mathcal{L}_{X_f}\omega$ is the Lie derivative of $\omega$ with respect to $X_f$. 

Let us define a bracket $\{\cdot,\cdot\}:(f,g)\in C^\infty(P)\times C^\infty(P)\mapsto \omega(X_f,X_g)\in C^\infty(P)$. 
This bracket is bilinear, antisymmetric, and, since $d\omega=0$, it obeys the {\it  Jacobi identity}, which makes $\{\cdot,\cdot\}$ into a {\it  Lie bracket}. Moreover, $\{\cdot,\cdot\}$ obeys the {\it  Leibniz rule}, i.e. $\{f,gh\}=\{f,g\}h+g\{f,h\}$ for all $f,g,h\in C^\infty(P)$. Mentioned properties turn $\{\cdot,\cdot\}$ into a so-called {\it  Poisson bracket}. It can be proved that
 $X_{\{g,f\}}=[X_f,X_g]$ (see \cite[p. 194]{AM78}).

Let us recall that $\mathfrak{g}$ stands for the Lie algebra of a Lie group $G$.
The {\it  fundamental vector field} of a Lie group action $\Phi:G\times P\rightarrow P$ related to $\xi\in\mathfrak{g}$ is the vector field on $P$ given by
\[
(\xi_P)_p:=\frac{d}{dt}\bigg|_{t=0}\Phi(\exp(t\xi),p),\quad \forall p\in P.
\]
Our convention in the definition of  fundamental vector fields gives rise to an anti-morphism of Lie algebras $\xi\in \mathfrak{g}\mapsto \xi_P\in \mathfrak{X}(P)$ (cf. \cite{CGM00}).
If $\Phi$ is known from context, we will write $gp$ instead of $\Phi(g,p)$ for every $g\in G$ and $p\in P$.
By the constant rank theorem \cite[p. 48]{AM78}, the orbits of $\Phi$ are immersed submanifolds in $P$. We also define 
\[
\Phi_g:\tilde{p}\in P\mapsto g\tilde{p}\in P,\qquad  \Phi^p:\tilde{g}\in G\mapsto \tilde{g}p\in P,\qquad \forall g\in G,\quad \forall p\in P.
\]
Each 
$\Phi_g$ is a diffeomorphism for every $g\in G$. 
The {\it  isotropy subgroup} of $\Phi$ at $p\!\in\!P$ is $G_p:=\{g\in G:gp=p\}\subset G$.
Let $G p$  stand for the orbit of $p\in P$ relative to $\Phi$, i.e. $G p:=\{gp\,:\,g\in G\}$. Then, $T_{\tilde{p}}Gp\!=\!\{(\xi_P)_{\tilde{p}}:\xi\in\mathfrak{g}\}$ for each $\tilde{p}\in Gp$.  

Recall that each $g\in G$ acts as a diffeomorphism on $G$ in the following manners: 
\[L_g:h\in  G\mapsto gh\in G,\qquad R_g:h\in G\mapsto hg\in G,\qquad I_g:h\in G\mapsto ghg^{-1}\in G.
\]
We hereafter assume that $G$ acts on $\mathfrak{g}$ via the {\it  adjoint action}, namely
\begin{equation}
\label{adaction}
{\rm Ad}:(g,\xi)\in G\times\mathfrak{g}\mapsto {\rm Ad}_g\xi\in\mathfrak{g},
\end{equation} 
where ${\rm Ad}_g\xi:=(T_eI_g)( \xi)$.
The fundamental vector field of the adjoint action related to $\xi\in \mathfrak{g}$ is given by
\[
(\xi_{\mathfrak{g}})_v=\frac{d}{dt}\bigg|_{t=0}{\rm Ad}_{\exp(t\xi)}(v)=[\xi,v]=:{\rm ad}_{\xi}v,\quad  \forall v\in\mathfrak{g},
\]
where $[\cdot,\cdot]$ denotes the Lie bracket in $\mathfrak{g}$. Note that $(\xi_\mathfrak{g})_v\in T_v\mathfrak{g}$ and ${\rm ad}_\xi v\in \mathfrak{g}$ may be defined to be equal because, for every finite-dimensional vector space $V$, there exists a natural isomorphism $v\in V\simeq D_v\in T_wV$, at each $w\in V$, identifying each $v\in V$ to the tangent vector at $w$ associated with the derivative at $w$ in the direction $v$.
Let $\mathcal{S}_{\xi}$ be the orbit of the {\it  adjoint action} passing through $\xi\in\mathfrak{g}$. Then,   $T_{\nu}\mathcal{S}_{\xi}=\{(\xi_{\mathfrak{g}})_\nu\,:\,\xi\in\mathfrak{g}\}$ for every $\nu\in \mathcal{S}_\xi$. 

The Lie group $G$ also acts on $\mathfrak{g}^*$ through the {\it  coadjoint action} 
$
{\rm Ad}^*:(g,\mu)\in G\times\mathfrak{g}^*\mapsto {\rm Ad}^*_{g^{-1}}\mu\in \mathfrak{g}^*,
$
where ${\rm Ad}_{g}^{*}$ is the transpose of ${\rm Ad}_{g}$, i.e. $\langle {\rm Ad}_{g}^{*}\mu, \xi\rangle = \langle \mu,{\rm Ad}_{g}\xi\rangle$ for all $\xi\in\mathfrak{g}$,  and where $\langle\cdot,\cdot\rangle$ denotes\footnote{This symbol has already been defined on page 9 with another similar meaning, but it does not lead to misunderstanding and our convention simplifies the notation.} the duality pairing between $\mathfrak{g}^{*}$ and $\mathfrak{g}$. Then, one has that
\begin{equation}\label{Dual}
(\xi_{\mathfrak{g}^*})_\mu=\frac{d}{dt}\bigg|_{t=0}{\rm Ad}_{\exp(-t\xi)}^*\mu=-\langle\mu, [\xi,\cdot ] \rangle=-{\rm ad}_{\xi}^*\mu,\quad \forall\mu\in\mathfrak{g}^*,
\end{equation}
where it is worth stressing that ${\rm ad}^*_\xi$ is defined to be ${\rm ad}^*_\xi(\vartheta):=\vartheta\circ {\rm ad}_\xi\in \mathfrak{g}^*$ for every $\vartheta\in \mathfrak{g}^*$.
Given the coadjoint orbit of $\mu\in\mathfrak{g}^*$, i.e.  $\mathcal{O}_{\mu}:=\{{\rm Ad}_{g^{-1}}^*\mu: g\in G\}
$, we have $
T_{\nu}\mathcal{O}_{\mu}=\{ (\xi_{\mathfrak{g}^*})_\nu : \xi\in\mathfrak{g}\}$  at every $\nu\in\mathcal{O}_{\mu}$. Then, $\xi_{\mathfrak{g}}$ and $\xi_{\mathfrak{g}^*}$ are related as follows
\[
\langle (\xi_{\mathfrak{g}^*})_\nu, v\rangle=\langle -{\rm ad}_\xi^*\nu,v\rangle=-\langle\nu,(\xi_{\mathfrak{g}})_v\rangle,\qquad \forall v\in \mathfrak{g}\simeq T^*_\nu \mathfrak{g}^*, \quad \forall \nu\in \mathfrak{g}^*\simeq T^*_v\mathfrak{g}.
\]

A Lie group action $\Phi:G\times P\rightarrow P$ is {\it Hamiltonian} if its fundamental vector fields are Hamiltonian relative to $\omega$.
An {\it  equivariant momentum map} for a Lie group action $\Phi:G\times P\rightarrow P$ is a map $\mathbf{J}:P\rightarrow\mathfrak{g}^*$ such that: \begin{enumerate}
     \item   $\mathbf{J}(g p)={\rm Ad}^*_{g^{-1}}(\mathbf{J}(p)),$ for all $g\in G$ and every $p\in P$. 
     \item  $(\iota_{\xi_P}\omega)_p=d\langle \mathbf{J}(p),\xi\rangle=(dJ_\xi)_p,$ for all $\xi\in\mathfrak{g}$, every $p\in P$, and $J_\xi:p\in P \mapsto \langle\mathbf{J}(p),\xi\rangle\in\mathbb{R}$. 
 \end{enumerate}
 We obtain that 2. gives that $\Phi$ is a Hamiltonian Lie group action and
\[
(\xi_PJ_\nu)(p)=\frac{d}{dt}\bigg|_{t=0}\!\!\!\!\langle {\bf J}(\exp(t\xi)p),\nu\rangle=\frac{d}{dt}\bigg|_{t=0}\!\!\!\!\langle {\rm Ad}^*_{\exp(-t\xi)}({\bf J}(p)),\nu\rangle=J_{[\nu,\xi]}(p),
\]
for all $\xi,\nu\in \!\mathfrak{g}$ and~$p\!\in\!P$. Then, $\{J_
\nu,J_\xi\}=J_{[\nu,\xi]}$. Hence, ${\bf J}$ gives rise to a Lie algebra morphism $\nu\in \mathfrak{g}\mapsto J_\nu\in C^\infty(P)$.

Let us go back to the naturally defined, and ubiquitous in physics, structures on a cotangent bundle $T^*Q$. 
A Lie group action $\Psi:G\times Q\rightarrow Q$ induces a new Lie group action $\Phi:(g,\alpha_q)\in G\times T^*Q\mapsto \Phi_g(\alpha_q)\in T^*Q$ such that
\[
\langle\Phi_{g}(\alpha_{q}),\ v_{g q}\rangle:=\langle \alpha_{q}, T_{gq}\Psi_{g^{-1}} (v_{g  q})\rangle,\qquad\forall q\in Q,  \quad \forall {v_{gq}} \in T_{gq}Q,
\]
the so-called {\it  cotangent lift} of $\Psi$. This notion is ubiquitous in geometric mechanics and it provides easily derivable momentum maps \cite[p. 283]{AM78}. Some additional details are given in the following proposition (see also \cite[p. 283]{AM78}).

\begin{proposition}
\label{LiftSymAc}
Every Lie group action $\Psi:G\times Q\rightarrow Q$ has a cotangent lift $\Phi:G\times T^*Q\rightarrow T^*Q$
admitting an equivariant momentum map $\mathbf{J}:T^*Q\rightarrow \mathfrak{g}^{*}$ such that
\begin{equation}\label{MomMap}
J_\xi(\alpha_{q})=:\langle {\bf J}(\alpha_q), \xi\rangle,\quad
J_\xi(\alpha_{q}):=\langle\alpha_{q},(\xi_{Q})_q\rangle,\qquad \forall \alpha_q\in T_q^* Q,\quad \forall q\in Q,\quad \forall \xi\in \mathfrak{g}.
\end{equation}
\end{proposition}

We hereafter assume that $\mu\in \mathfrak{g}^*$ is a regular value of ${\bf J}$. Hence, ${\bf J}^{-1}(\mu)$ is a submanifold of $P$ and $T_p({\bf J}^{-1}(\mu))=\ker(T_p {\bf J})$ for every $p\in {\bf J}^{-1}(\mu)$.
\begin{proposition} 
\label{2.2}
If $p\in \mathbf{J}^{-1}(\mu)$ for a regular $\mu\in \mathfrak{g}^*$ and $G_\mu$ is the isotropy group of $\mu$ relative to the coadjoint action of $G$, then: 
\begin{enumerate}
    \item $T_{p}(G_{\mu} p)=T_{p}(G p)\cap T_{p}(\mathbf{J}^{-1}(\mu))$,
    \item  $T_{{p}}({\bf J}^{-1}(\mu))=(T_{p}Gp)^{\perp_{\omega}}$.
\end{enumerate}
\end{proposition}

Let us enunciate the Marsden--Weinstein theorem (see \cite[p. 300]{AM78} and the original work by Marsden and Weinstein \cite{MW74}).

\begin{theorem} Let $\Phi:G\times P\rightarrow P$ be a Hamiltonian Lie group action of $G$ on the symplectic manifold $(P,\omega)$ admitting an equivariant momentum map ${\bf J}:P\rightarrow \mathfrak{g}^*$. Assume that $\mu\in \mathfrak{g}^*$ is a regular point of ${\bf J}$ and $G_\mu$, the isotropy group of $\mu$ relative to the coadjoint action, acts freely and properly on ${\bf J}^{-1}(\mu)$. Let $\iota_{\mu}:{\bf J}^{-1}(\mu)\rightarrow P$ denote a natural embedding and let $\pi_\mu:{\bf J}^{-1}(\mu)\rightarrow {\bf J}^{-1}(\mu)/G_\mu=:P_\mu$ be the canonical projection onto the space of orbits of $G_\mu$ acting on ${\bf J}^{-1}(\mu)$.  There exists a unique symplectic structure $\omega_{\mu}$  on $P_{\mu}$ such that $\pi_{\mu}^*\omega_{\mu}=\iota_{\mu}^{*}\omega$.
\end{theorem}

\begin{definition}
A {\it $G$-invariant Hamiltonian system} is a 5-tuple $(P,\omega,h,\Phi,{\bf J})$, where $\Phi$ is a Lie group action of $G$ on $P$ with an equivariant momentum map ${\bf J}$, and $h:\mathbb{R}\times P\rightarrow \mathbb{R}$ is a real $t$-dependent function on $P$ satisfying $h(t,\Phi(g,p))=h(t,p)$ for every $g\in G$, $t\in \mathbb{R}$, and $p\in P$. 
\end{definition}

Note that  $h:\mathbb{R}\times P\rightarrow \mathbb{R}$ gives rise to a $t$-dependent vector field on $P$ of the form $X_h:\mathbb{R}\times P\rightarrow TP$ such that each vector field $X_{h_t}:p\in P\mapsto X_h(t,p)\in TP$, with $t\in \mathbb{R}$, is the Hamiltonian vector field of $h_t:p\in P\mapsto h(t,p)\in \mathbb{R}$. Then, a {\it particular solution, $p(t)$, to $(P,\omega,h,\Phi,{\bf J})$} is called a particular solution of the non-autonomous system of differential equations
\[
\frac{dp}{dt}=X_{h_t}(p)=X_h(t,p),\qquad \forall (t,p)\in \mathbb{R}\times P.
\]

From now on, $(P,\omega,h,\Phi,\mathbf{J})$ will always stand for a $G$-invariant Hamiltonian system. Proposition \ref{Prop:EvPhiT} analyses the evolution of ${\bf J}:P\rightarrow \mathfrak{g}^{*}$ under the dynamics of the $t$-dependent vector field $X_{h}$ determined by a $G$-invariant Hamiltonian system  $(P,\omega,h,\Phi,{\bf J})$.  In particular, let us briefly prove that   $\mathbf{J}:P\rightarrow \mathfrak{g}^{*}$ is conserved for the dynamics of $X_{h}$, i.e. the flow, $F:\mathbb{R}\times P\rightarrow P$, of the $t$-dependent vector field $X_h$ leaves  $\mathbf{J}$ invariant and, consequently, ${\bf J}\circ F_t={\bf J}$ for $F_t:p\in P\mapsto F(t,p)\in P$ and every $t\in\mathbb{R}$. Our proof is just an analogue of the $t$-independent case that can be found in any standard reference \cite[p. 277]{AM78}.

\begin{proposition}\label{Prop:EvPhiT} 
Let $(P,\omega,h,\Phi,\mathbf{J})$ be a $G$-invariant Hamiltonian system. Then, $\mathbf{J}$ is invariant relative to the evolution of $h$, i.e. if $F:\mathbb{R}\times P\rightarrow P$ is the flow of the $t$-dependent vector field on $P$ given by $X_h:(t,p)\in \mathbb{R}\times P\mapsto X_h(t,p)\in TP$, then
\[
\mathbf{J}(F(t,p))=\mathbf{J}(p),\qquad \forall p\in P, \quad \forall t\in \mathbb{R}.
\]
\end{proposition}
\begin{proof}  Let us define $F_t:p\in P\mapsto F(t,p)\in P$ for every $t\in \mathbb{R}$. On the one hand,
\[
\frac{d}{dt}J_{\xi}(F_t)=(X_{h_t}J_{\xi})\circ F_t=\{J_{\xi},h_t\}\circ F_t=(-X_{J_{\xi}}h_t)\circ F_t=-(\xi_Ph_t)\circ F_t=0, \forall \xi\in \mathfrak{g}, \forall t\in \mathbb{R},
\]
where the last equality stems from the fact that each $h_t$, for $t\in \mathbb{R}$, is invariant by assumption relative to the fundamental vector fields of  the action of $G$ on $P$, namely, the vector fields $\xi_P$ with $\xi\in\mathfrak{g}$. Since the $J_{\xi}$ is invariant relative to the dynamics induced by $h$ for every $\xi\in \mathfrak{g}$, we get that $\mathbf{J}$ is invariant relative to the evolution in time of the Hamiltonian system determined by $h$.
\end{proof}

The ${G}$-invariance property of $h$ also yields that $F$  induces canonically a Hamiltonian flow on the reduced phase space $P_{\mu}=\mathbf{J}^{-1}(\mu)/G_{\mu}$ associated with a Hamiltonian function $k_{\mu}:\mathbb{R}\times P_{\mu}\rightarrow \mathbb{R}$ defined in a unique way via the equation $k_{\mu}(t,\pi_\mu(p))=h(t, p)$ for every $p\in {\bf J}^{-1}(\mu)$, the referred to as {\it reduced Hamiltonian}. The proof of this fact is a straightforward generalisation of its $t$-independent proof (cf \cite{AM78,MW74}). Let us prove certain facts on the geometry of the regular elements of ${\bf J}$ for $(P,\omega,h,\Phi,{\bf J})$.

\begin{theorem} If $\mu$ is a regular value for the momentum map ${\bf J}$
of $(P,\omega,h,\Phi,{\bf J})$, then every $\mu'$ belonging to the coadjoint orbit, $\mathcal{O}_\mu$, of $\mu\in \mathfrak{g}^*$ is also a regular value. If $G_\mu$ acts properly and freely in ${\bf J}^{-1}(\mu)$, then $G_{\mu'}$ acts also freely and properly on ${\bf J}^{-1}(\mu')$ for every $\mu'\in \mathcal{O}_\mu$. Finally, ${\bf J}^{-1}(\mathcal{O}_\mu)$ is a submanifold of $P$. 
\end{theorem}
\begin{proof}
If $\mu$ is a regular point of {\bf J}, then $T{\bf J}$ is a surjection on the points of ${\bf J}^{-1}(\mu)$. The equivariance of ${\bf J}$ yields that, for any $g\in G$ and $p\in {\bf J}^{-1}(\mu)$, one has that ${\bf J}(gp)={\rm Ad}^*_{g^{-1}}({\bf J}(p))$. Let us set $\mu':={\rm Ad}_{g^{-1}}^*\mu$. Hence, if $p\in {\bf J}^{-1}(\mu)$, then $gp\in {\bf J}^{-1} (\mu')$. Since $\Phi_g$ is a diffeomorphism, it follows that  
\[{\bf J}^{-1}({\rm Ad}_{g^{-1}}^*\mu)=\Phi_g({\bf J}^{-1}(\mu)),\qquad \forall g\in G,\quad\forall \mu\in {\bf J}(P).
\]
Moreover, $T_{gp}{\bf J}={\rm Ad}_{g^{-1}}^*T_{p}{\bf J}$ for every $p\in \mathbf{J}^{-1}(\mu)$ and $g\in G$. Then,  $T{\bf J}$ is a surjection on ${\bf J}^{-1}({\rm Ad}_{g^{-1}}^*\mu)$ for every $g\in G$. 

Note that $G_{{\rm Ad}^*_{g^{-1}}\mu}=I_{g}G_\mu $ for every $g\in G$ and $\mu\in {\bf J}(P)$.  Moreover, if $\Phi:G_\mu\times {\bf J}^{-1}(\mu)\rightarrow {\bf J}^{-1}(\mu)$ is free and proper, by the equivariance of $\Phi$, it follows that $\Phi:G_{\mu'}\times {\bf J}^{-1}(\mu')\rightarrow {\bf J}^{-1}(\mu')$ is free and proper also for $\mu'\in \mathcal{O}_{\mu}$. 

To prove that ${\bf J}^{-1}(\mathcal{O}_\mu)$ is a submanifold of $P$, we recall that if $f:M\rightarrow N$, $S\subset N$ is a submanifold of the manifold $N$  and ${\rm Im}\,T_{p}f+ T_{s}S=T_sN$ for every $s\in S$ and $p\in f^{-1}(s)$, we say that $f$ is {\it transversal} to $S$. Then, $f^{-1}(S)$ is a submanifold of $M$ (see \cite[p. 49]{AM78}). Since $\mu$ is a regular point of ${\bf J}$, one has that ${\rm Im}\,T_p {\bf J}=T_{{\bf J}(p)}\mathfrak{g}^*$ for every $p\in {\bf J}^{-1}(\mu)$. Consequently, ${\rm Im}\,T_p{\bf J}+T_{{\bf J}(p)}\mathcal{O}_\mu=T_{{\bf J}(p)}\mathfrak{g}^*$ for every $p\in {\bf J}^{-1}(\mathcal{O}_\mu)$. Therefore, ${\bf J}$ is transversal to $\mathcal{O}_\mu$ and ${\bf J}^{-1}(\mathcal{O}_\mu)$ is a submanifold of $P$. 
\end{proof}
\section{Relative equilibrium points}

Let us extend Poincar\'{e}'s terminology of a {\it relative equilibrium point} (see \cite[p. 306]{AM78}) for a $t$-independent Hamiltonian function to the realm of $t$-dependent Hamiltonian systems on symplectic manifolds.

\begin{definition} 
\label{RelEqPoint}
A {\it  relative equilibrium point} for  $(P,\omega,h,\Phi,{\bf J})$ is a point $z_e\in P$ such that there exists a curve $\xi(t)$ in $\mathfrak{g}$ so that
\begin{equation}\label{def:RelEquPoi}
(X_{h_t})_{z_e}=(\xi(t)_P)_{z_e},\qquad \forall t\in \mathbb{R}.
\end{equation}
\end{definition}

Definition (\ref{RelEqPoint}) reduces to the standard relative equilibrium point for autonomous systems. The following proposition explains more carefully why $z_e$ can still be called a relative equilibrium point.

\begin{proposition}\label{Prop:SingP} Every solution, $p(t)$, to $(P,\omega,h,\Phi,{\bf J})$ passing through a relative equilibrium point $z_e\in P$ with $\mu_e:={\bf J}(z_e)$, namely $p(t_0)=z_e$ for some $t_0\in \mathbb{R}$, projects onto the point $\pi_{\mu_e} (z_e)$, i.e. $\pi_{\mu_e}(p(t))=\pi_{\mu_e}(z_e)$ for every $t\in\mathbb{R}$. 
\end{proposition}
\begin{proof}
By Proposition \ref{Prop:EvPhiT}, every solution $p(t)$ to the Hamilton equations of $h$ is fully contained within a certain submanifold $\mathbf{J}^{-1}(\mu)$. Then, $p(t)$ projects, via $\pi_{\mu_e}$, onto a curve in $P_{\mu_e}:=\mathbf{J}^{-1}(\mu_e)/G_{\mu_e}$, where $G_{\mu_e}$ is the isotropy subgroup of $\mu_e$ relative to the coadjoint action. Such a curve is a solution to the Hamiltonian system $(P_{\mu_e},\omega_{\mu_e},k_{\mu_e})$, where $k_{\mu_e}:\mathbb{R}\times P_{\mu_e}\rightarrow \mathbb{R}$ is the only $t$-dependent function on $P_{\mu_e}$ such that $k_{\mu_e}(t,  \pi_{\mu_e}(p))=h(t,p)$ for every $p\in {\bf J}^{-1}(\mu_e)$ and $t\in\mathbb{R}$. Since $z_e$ is a relative equilibrium point, it turns out that 
\[
0=T{\bf J}(X_{h_{t}})_{z_e}=T\mathbf{J}(\xi(t)_{P})_{z_e}=(\xi(t))_{\mathfrak{g}^*}(\mu),\qquad \forall t\in \mathbb{R},
\]
for some curve $\xi(t)$ in $\mathfrak{g}$. Hence, $\xi(t)\in\mathfrak{g}_{\mu_e}$ for every $t\in \mathbb{R}$.

Note that $\pi_{\mu_e}(p(t))$ is the integral curve to the $t$-dependent vector field $Y_{\mu_e}$ on $P_{\mu_e}$ given by the $t$-parametric family of vector fields on $P_{\mu_e}$ of the form $(Y_{\mu_e})_t:=\pi_{\mu_e*}(X_{h_t})$ for every $t\in \mathbb{R}$. Since $X_{h_t}=\xi(t)_P$, for a certain curve $\xi(t)$ contained in $\mathfrak{g}_{\mu_e}$, then $((Y_{\mu_e})_{t})_{\pi_{\mu_e}(z_e)}=T_{z_e}\pi_{{\mu_e}}(\xi(t)_P)_{z_e}=0$ for every $t\in \mathbb{R}$. As a consequence, $\pi_{\mu_e}(z_e)$ is an equilibrium point of $Y_{\mu_e}$ and the integral curve of the $t$-dependent vector field $Y_{\mu_e}$ passing through $\pi_{\mu_e}(z_e)$ is $\pi_{\mu_e}(z_e)$. Hence, $\pi_{\mu_e}(p(t))=\pi_{\mu_e}(z_e)$ for every $t\in \mathbb{R}$ and $p(t)\in \pi^{-1}_{\mu_e}(z_e)$ for every $t\in \mathbb{R}$. Then, the projection of every solution passing through $z_e$ is just the stability point $\pi_{\mu_e}(z_e)$ of the reduced Hamiltonian system related to $Y_{\mu_e}$ on $P_{\mu_e}$. 
\end{proof}

Proposition \ref{Prop:SingP} yields that every solution passing through a relative equilibrium point $z_e$ with ${\bf J}(z_e)=\mu_e$ satisfies that $p(t)=g(t)z_e$ for a certain curve $g(t)$ in $G_{\mu_e}$. Let us show that the converse is also true. 

\begin{proposition} If every solution $p(t)$ to $(P,\omega,h,\Phi,{\bf J})$  passing through a point $z_e\in P$, with $\mu_e:={\bf J}(z_e)$, projects onto $\pi_{\mu_e} (z_e)$, then $z_e$ is a relative equilibrium point. 
\end{proposition}
\begin{proof} Let $p(t)$ be the solution to $(P,\omega,h,\Phi,{\bf J})$ passing through $z_e$ at $t=t_0$. By our assumptions, $\pi_{\mu_e}(p(t))$ projects onto $\pi_{\mu_e}(z_e)$. Consequently, there exists a curve $g(t)$ in $G_{\mu_e}$ such that $p(t)=\Phi(g(t),p(t_0))$ and $g(t_0)=e$. Therefore,
\[
 (X_{h_{t_0}})_{z_e}=\frac{dp}{dt}(t_0)= \frac{d}{dt}\bigg|_{t=t_0} (\Phi(g(t),z_e))=T_{e}\Phi_{z_e}\left(\frac{dg}{dt}(t_0)\right)=(\nu(t_0))_{P}({z_e}),
 \]
for a certain $\nu(t_0)\in \mathfrak{g}_{\mu_e}$. Since the above holds for every $t_0\in \mathbb{R}$, we obtain that $z_e$ is a relative equilibrium point.
\end{proof}

Note that if $p(t)$ is a solution to $(P,\omega,h,\Phi,{\bf J})$ and $p(t)=g(t)p$, Proposition \ref{Prop:EvPhiT} ensures that ${\bf J}(p(t))={\bf J}(p)$. Hence, the action of $g(t)$ leaves invariant the value of ${\bf J}(p)$ and it belongs to $G_{\mu_e}$ for $\mu_e={\bf J}(p)$. 
From previous results, we have the following corollary.
\begin{corollary}\label{Cor:Eq} The following two conditions are equivalent:
\begin{itemize}
    \item The point $z_e\in P$ is a relative equilibrium point of $(P,\omega,h,\Phi,{\bf J})$,
    \item Every particular solution to $(P,\omega,h,\Phi,{\bf J})$ passing through $z_e\in P$ is of the form $p(t)=g(t)z_e$ for a curve $g(t)$ in $G$.
\end{itemize}
\end{corollary}

It is remarkable that, in $t$-dependent systems, the Hamiltonian needs not be a constant of the motion since
\begin{equation}\label{Eq:EvoEnergy}
\frac{dh}{dt}=\frac{\partial h}{\partial t}+\{h,h\}=\frac{\partial h}{\partial t}.
\end{equation}
Meanwhile, Corollary \ref{Cor:Eq} ensures that for particular solutions $p(t)=g(t)z_e$, it follows that $h(t,p(t))=h(t,z_e)$. Despite that, $h$ need not be a constant of the motion  along solutions to $h$ even when passing through relative equilibrium points. It is remarkable that, since $h$ is not a constant of the motion, the analysis of the stability of solutions of the reduced Hamiltonian systems  $k_{\mu_e}$ on $P_{\mu_e}$ will be much more complicated. Indeed, as $k_{\mu_e}$ will not be in general autonomous, much of the procedures given in standard stability analysis must be substituted by more general approaches (cf. \cite{MS88}).

The following proposition allows us to characterise relative equilibrium points more easily than through previous methods.

\begin{theorem}\label{RelativeEquilibrum}
{\bf (Time-Dependent Relative Equilibrium Theorem)} A point $z_{e}\in P$ is a relative equilibrium point for $(P,\omega,h,\Phi,{\bf J})$ if and only if there exists a curve $\xi(t)$ in $\mathfrak{g}$   such that $z_{e}$  is a critical point of $h_{\xi,t}:P\rightarrow \mathbb{R}$ given by
\[h_{\xi,t}:=h_t-[J_{\xi(t)}-\langle \mu_{e},\xi(t)\rangle]=h_t-\langle \mathbf{J}-\mu_{e},\ \xi(t)\rangle
\] 
for every $t\in \mathbb{R}$ and $\mu_e:={\bf J}(z_e)$.
 \end{theorem}
\begin{proof} Assume first that $z_{e}$ is a relative equilibrium point. The definition of the momentum map and Corollary \ref{Cor:Eq} yield $(X_{h_t})_{z_{e}}\!-\!(X_{J_{\xi(t)}})_{z_{e}}=0$ for every $t\in \mathbb{R}$. Since $P$ is symplectic, the latter is equivalent to $z_{e}$ being a critical point of $h_t-J_{\xi(t)}$  for every $t\in \mathbb{R}$, which is the same as being a critical point of $h_{\xi,t}$ for every $t\in \mathbb{R}$, namely $(dh_{\xi,t})_{z_e}=0$.

Conversely, assume $z_{e}$ is a critical point of $h_{\xi,t}$, then $z_{e}$ is a stationary point of the dynamical system $X_{h_t-J_{\xi(t)}}$ for every $t\in \mathbb{R}$. Hence, the evolution of every particular solution of $X_{h}$ passing through $z_e$ at time $t_0$ is of the form $g(t)z_e$ for a certain curve in $G$ with $g(t_0)=e$ and, in view of Corollary \ref{Cor:Eq}, one has that $z_e$ becomes a relative equilibrium point.
\end{proof}
In view of Theorem \ref{RelativeEquilibrum}, to find relative equilibrium points, one can consider the functions $h^e_t:(p,\xi)\in P\times\mathfrak{g}\mapsto h_t-\langle \mathbf{J}-\mu_{e},\ \xi \rangle\in \mathbb{R}$, for every $t\in \mathbb{R}$, and to look for elements $z_e\in P$ such that $(z_e,\xi(t))$ is a critical point of $h^e_t$ for each $t\in \mathbb{R}$ and a certain curve $\xi(t)$ in  $\mathfrak{g}$. Evidently, $\xi(t)$ plays here the role of a $t$-dependent Lagrange multiplier. Note that the term $\langle {\bf J}-\mu_e,\xi\rangle$ in $h^e_t$ ensures that the described relative equilibrium points  belong to ${\bf J}^{-1}(\mu_e)$.

\section{Foliated Lie systems and relative equilibrium submanifolds}

This section shows that the set of relative equilibrium points for a $G$-invariant Hamiltonian system $(P,\omega,h,\Phi,{\bf J})$ is given by a union of immersed submanifolds. Moreover, we also prove that the restriction of the original $t$-dependent Hamiltonian system to such immersed submanifolds can be described via a foliated Lie system  \cite{CGM00} assuming a certain condition on the Lie algebra of fundamental vector fields of the action of $G$ on $P$.  

\begin{proposition}\label{Prop:REQonIM} If $z_e$ is a relative equilibrium point of $(P,\omega,h,\Phi,{\bf J})$, then $\mathcal{O}_{z_e}:=Gz_e$ is an immersed submanifold of $P$ consisting of relative equilibrium points.
\end{proposition}
\begin{proof} Since $z_e$ is a relative equilibrium point, every solution passing through $z_e$ is of the form $z(t)=g(t)z_e$ for a certain curve $g(t)$ in $G$. Since  $h(t, \Phi_g(p))=h(t,p)$ for every $t\in\mathbb{R}$ and $p\in P$, and also $\Phi_g^*\omega=\omega$ for every $g\in G$, one obtains that 
\begin{multline*}
\iota_{X_{h_t}}\omega=dh_t\Rightarrow (\iota_Y\iota_{\Phi_{g*}X_{h_t}}\omega)(gp)=[(\Phi^*_g\omega)(X_{h_t},\Phi_{g^{-1}*}Y)](p)\\=\omega(X_{h_t},\Phi_{g^{-1}*}Y)(p)=\langle dh_t,\Phi_{g^{-1}*}Y\rangle(p)=\langle d\Phi_{g^{-1}}^*h_t,Y\rangle(gp)=\langle dh_t,Y\rangle(gp),
\end{multline*}
for every $Y\in \mathfrak{X}(P)$, $g\in G$, $p\in P$ and $t\in\mathbb{R}$.
Therefore, $\Phi_{g*}X_{h_t}=X_{h_t}$ for every $t\in \mathbb{R}$. Hence, every solution $z'(t)$ passing through $gz_e$ is such that $z(t):=g^{-1}z'(t)$ is a solution to $X_{h_t}$ passing through $z_e$. Thus, $z'(t)=gz(t)=gg(t)g^{-1}gz_e$. 
In other words, $gz_e$ is a relative equilibrium point for $(P,\omega,h,\Phi,{\bf J})$. Since $Gz_e$ is an immersed submanifold of $P$ (see \cite{Bo05}), our proposition follows.
\end{proof}

A {\it foliated Lie system} \cite{CGM00} on a manifold $P$ is a first-order system of differential equations taking the form
\[
\frac{dp}{dt}=X(t,p),\qquad \forall t\in \mathbb{R},\qquad \forall p\in P,
\]
so that 
\begin{equation}\label{eq:FLS1}
X(t,p)=\sum_{\alpha=1}^rg_\alpha (t,p)X_\alpha(p), \qquad \forall t\in \mathbb{R},\quad \forall p\in P,
\end{equation}
where $X_1,\ldots,X_r$ span an $r$-dimensional Lie algebra of vector fields, i.e.
\[
[X_\alpha,X_\beta]=\sum_{\gamma=1}^rc_{\alpha\beta}^\gamma X_\gamma,\qquad \alpha,\beta=1,\ldots,r,
\]
for certain constants $c_{\alpha\beta}^\gamma$, and the functions $g_{\alpha,t}:p\in P\mapsto g_\alpha(t,p)\in \mathbb{R}$, for every $t\in \mathbb{R}$ and $\alpha=1,\ldots,r$, are first integrals of $X_1,\ldots,X_r$. The Lie algebra $\langle X_1,\dots,X_r\rangle$ is called a {\it Vessiot--Guldberg Lie algebra} of $X$ \cite{LS20}.

Let us show how foliated Lie systems occur in the study of relative equilibrium points for $G$-invariant Hamiltonian systems.

\begin{theorem} Let $z_e$ be a relative equilibrium point for $(P,\omega,h,\Phi,{\bf J})$ and let $\mu_e:={\bf J}(z_e)$. Assume that $G_{\mu_e}$ is abelian. Then,  $X_{h_t}$ is tangent to $\mathcal{O}_{z_e}:=G z_e$ for every fixed $t\in \mathbb{R}$, and $X_{h}$ gives rise, by restriction, to a $t$-dependent vector field $X_h|_{\mathcal{O}_{z_e}}$ on $\mathcal{O}_{z_e}$. Moreover, $X_h|_{\mathcal{O}_{z_e}}$ becomes a foliated Lie system with an abelian Vessiot--Guldberg Lie algebra of dimension equal to $\dim\mathfrak{g}_{\mu_e}$.
\end{theorem}
\begin{proof} From now on in this proof, we assume that $z_e'$ belongs to $\mathcal{O}_{z_e}$ and $\mu_e':={\bf J}(z_e')$. Let us prove that $X_h|_{\mathcal{O}_{z_e}}$ exists and it can be written in the form (\ref{eq:FLS1}) for certain functions $g_\alpha$, with $\alpha=1,\ldots,r$, that depend only on time on the submanifolds of the form $G_{\mu'_e}z'_e$, where $G_{\mu'_e}$ is the isotropy subgroup of the coadjoint action of $G$ at $\mu'_e$, and certain vector fields tangent to $\mathcal{O}_{z_e}$ closing on a finite-dimensional Lie algebra of vector fields. This shows that $X_h|_{\mathcal{O}_{z_e}}$ is a foliated Lie system.

Proposition \ref{Prop:REQonIM} ensures  that $z'_e$ is a relative equilibrium point. Then, every integral curve to $X_{h}$ passing through $z'_e$ takes the form $z(t)=g(t)z'_e$ for a certain curve $g(t)$ in $G$.  This shows that $X_h$ is tangent to $\mathcal{O}_{z_e}$ and it can be restricted to it. Proposition \ref{Prop:EvPhiT} yields that ${\bf J}$ is constant on integral curves of $X_{h}$. Consequently, the integral curves of $X_{h}$ passing through $z'_e$ are contained in ${\bf J}^{-1}(\mu'_e)$. Assume that $z(t_0)=z_e'$ and then $z(t)=g(t)z_e'$ for a curve $g(t)$ in $G$ with $g(t_0)=e$. Hence,
\[
0=\frac{d}{dt}\bigg|_{t=t_0}{\bf J}(z(t))=\frac{d}{dt}\bigg|_{t=t_0}{\bf J}(g(t)z'_e)=
\frac{d}{dt}\bigg|_{t=t_0}{\rm Ad}_{g(t)^{-1}}^*({\bf J}(z'_e))=[\xi(t_0)]_{\mathfrak{g}^*}(\mu'_e),\]
where $dg(t)/dt|_{t=t_0}=\xi(t_0)$. Therefore, $\xi(t_0)\in \mathfrak{g}_{\mu'_e}$.  

Let $\{\xi_1^{\mu_e},\ldots,\xi^{\mu_e}_r\}$ be a basis for $\mathfrak{g}_{\mu_e}$. By our initial assumptions, $G_{\mu_e}$ is abelian and thus $\mathfrak{g}_{\mu_e}$ is abelian too. Define the vector fields on $\mathcal{O}_{z_e}$ of the form $Y_\alpha(gz_e):=[T_{z_e}\Phi_g(\xi^{\mu_e}_\alpha)_P](z_e)$ for $\alpha=1,\ldots,r$. Note that each $Y_\alpha$ is well defined because $G_{z_e}$, the isotropy group of $G$ acting on $z_e$, is a subgroup of $G_{\mu_e}$, which acts freely on ${\bf J}^{-1}(
\mu_e)$, and then $gz_e=g'z_e$ implies $g^{-1}g'\in G_{z_e}=\{e\}$ and $g'=g$. Since the action of $G_{\mu_e}$ is assumed to be free on ${\bf J}^{-1}(\mu_e)$, the tangent vectors $Y_1(z_e),\ldots,Y_r(z_e)$ are linearly independent. Since $Y_\alpha(gz_e)=T_{z_e}\Phi_{g}[Y_\alpha(z_e)]$ for every $g\in G$, one obtains that $Y_1\wedge\ldots \wedge Y_r\neq 0$ on $\mathcal{O}_{z_e}$. Since $\mathfrak{g}_{\mu_e}$ is abelian, for every $g_{\mu_e}\in G_{\mu_e}$, one has 
\begin{equation}\label{Eq:Mid}
Y_\alpha(g g_{\mu_e} z_e)=T_{g_{\mu_e}z_e}\Phi_g\circ T_{z_e}\Phi_{g_{\mu_e}}[(\xi^{\mu_e}_\alpha)_P(z_e)]=T_{g_{\mu_e} z_e}\Phi_{g}(\xi^{\mu_e}_\alpha)_P(g_{\mu_e} z_e)=({\rm Ad}_g(\xi^{\mu_e}_\alpha))_P(gg_{\mu_e} z_e),
\end{equation}
for $\alpha=1,\ldots,r.$ Note indeed that ${\rm Ad}_g(\xi^{\mu_e}_\alpha)$, with $\alpha=1,\ldots,r$, is a basis of the Lie algebra $\mathfrak{g}_{{\bf J}(g g_{\mu_e} z_e)}$. Moreover,
  $X_{h}(t,z)=\sum_{\alpha=1}^rf_\alpha(t,z)Y_\alpha(z)$  on every  $z\in G_{\mu'_e}z'_e$ for a unique set of functions $f_1(t,z),\ldots,f_r(t,z)$. If $z_e':=gz_e$ and since  $G_{\mu_e}$ is abelian, then $G_{\mu'_e}=gG_{\mu_e}g^{-1}$ is abelian too. Then, $T_{z'_e}\Phi_{g_{\mu'_e} }({\rm Ad}_{g}(\xi^{\mu_e}_\alpha))_{P}( z'_e)=({\rm Ad}_{g}(\xi^{\mu_e}_\alpha))_{P}( g_{\mu'_e} z'_e)$.   By using (\ref{Eq:Mid}), we obtain
\begin{multline}\label{Eq:Mid2}
X_{h}(t,g_{\mu'_e} z'_e)=T_{z_e'}\Phi_{g_{\mu'_e}}X_{h}(t,z'_e)=\sum_{\alpha=1}^rf_\alpha(t,z'_e)T_{z'_e}\Phi_{g_{\mu'_e} }({\rm Ad}_{g}(\xi^{\mu_e}_\alpha))_{P}( z'_e)\\=\sum_{\alpha=1}^rf_\alpha(t,z'_e)({\rm Ad}_{g}(\xi^{\mu_e}_\alpha))_{P}( g_{\mu'_e} z'_e)=\sum_{\alpha=1}^rf_\alpha(t,z_e')Y_\alpha(g_{\mu_e'}z_e'),
\end{multline}
for every $g_{\mu_e'}\in G_{\mu_e'}$. In particular, we have used in the last equality that $g^{-1}G_{\mu'_e}g=G_{\mu_e}$ and then $g^{-1}g_{\mu_e'}g=g'_{\mu_e}$ for some $g'_{\mu_e}\in G_{\mu_e}$. Thus, 
\begin{multline*}
Y_\alpha(g_{\mu_e'}z_e')=Y(g_{\mu_e'}gz_e)=Y(gg^{-1}g_{\mu_e'}gz_e)=Y(gg'_{\mu_e}z_e)=({\rm Ad}_g(\xi^{\mu_e}_\alpha))_P(gg'_{\mu_e}z_e)\\=({\rm Ad}_g(\xi^{\mu_e}_\alpha))_P(g_{\mu'_e}gz_e)=({\rm Ad}_g(\xi^{\mu_e}_\alpha))_P(g_{\mu'_e}z'_e).
\end{multline*}
From (\ref{Eq:Mid2}) and the fact that $X_h(t,g_{\mu_e'}z')=\sum_{\alpha=1}^rf_\alpha(t,g_{\mu_e'}z')Y_\alpha(g_{\mu_e'}z')$, one has that $f_\alpha(t,z'_e)=f_\alpha(t,g_{\mu_e'} z'_e)$ for every $g_{\mu'_e}\in G_{\mu'_e}$ and $\alpha=1,\ldots,r$. Consequently, one obtains that
\[
X_{h}(t,z)=\sum_{\alpha=1}^rf_\alpha(t,z)Y_\alpha(z),\qquad \forall z\in \mathcal{O}_{z_e},\quad \forall t\in \mathbb{R},
\]
for some functions $f_1,\ldots,f_r$ on $\mathbb{R}\times\mathcal{O}_{z_e}$ whose values on  each subset $G_{\mu'_e}z'_e$, with $z'_e\in \mathcal{O}_{z_e}$ and $\mu'_e={\bf J}(z_e')$, depend only on time. The vector fields $Y_1,\ldots,Y_r$ are tangent to the submanifolds $G_{\mu'_e} z_e'$, where they close an abelian Lie algebra.   Since the functions $f_1,\ldots,f_r$ are just $t$-dependent on the submanifolds $G_{\mu_e'}z_e'$, they become first integrals of the vector fields in $\langle Y_1,\ldots,Y_r\rangle$. Therefore, $X_{h}|_{\mathcal{O}_z}$ becomes a foliated Lie system with an abelian Vessiot--Guldberg Lie algebra isomorphic to $\mathfrak{g}_{\mu_e}$.
\end{proof}

\section{Stability on the reduced space}
Theorem \ref{RelativeEquilibrum} characterises the relative equilibrium points of $G$-invariant Hamiltonian systems as the extrema of the Hamiltonian subject to the constraint of the constant momentum map. Then,  $h_{\xi,t}:=h_t-\langle \mathbf{J}-\mu_{e},\xi(t)\rangle$ is to be optimised and $\xi(t)\in \mathfrak{g}$ is a Lagrange multiplier depending on time. 

The study of the stability of equilibrium points in ${\bf J}^{-1}(\mu_e)/G_{\mu_e}$ for non-autonomous Hamiltonian systems requires the use of $t$-dependent Lyapunov stability analysis. This is more complicated than studying the stability of autonomous Hamiltonian systems, which frequently relies on searching a minimum for the Hamiltonian of the system   \cite{MS88}, although this condition is not necessary \cite[p. 207]{AM78}. To tackle the study of non-autonomous Hamiltonians, we will use Theorem \ref{Th:BasicTheoremOfLyapunov} and a more general approach, which easily retrieves the standard results used in the energy-momentum method for analysing the stability of reduced autonomous Hamiltonian systems.

Let $z_e$ be a relative equilibrium point of $(P,\omega,h,{\bf J},\Phi)$. Let us analyse the function $h_{z_e}:\mathbb{R}\times P\rightarrow \mathbb{R}$ given by
\[
h_{z_e}(t,z):=h(t,z)-h(t,z_e),\qquad \forall (t,z)\in \mathbb{R}\times P.
\]
Then, $h_{z_e}(t,z_e)=0$ for every $t\in \mathbb{R}$. If $z(t)$ is the particular solution to our $G$-invariant Hamiltonian system $(P,\omega,h,{\bf J},\Phi)$ with initial condition $z$ at $t_0$, then
\[
\frac{d}{dt}\bigg|_{t=t_0}h_{z_e}(t,z(t)):=\frac{d}{dt}\bigg|_{t=t_0}h(t,z(t))-\frac{d}{dt}\bigg|_{t=t_0}h(t,z_e).
\]
Recall that the time derivative of a Hamiltonian function $h$ along the solutions of its Hamilton equations is given by
\[
\frac{dh}{dt}=\frac{\partial h}{\partial t}+ \{h_t,h_t\}=\frac{\partial h}{\partial t}.
\]
Thus,
\[
\frac{d}{dt}\bigg|_{t=t_0}h_{z_e}(t,z(t)):=\frac{\partial h}{\partial t}(t_0,z)-\frac{\partial h}{\partial t}(t_0,z_e)=\frac{\partial h_{z_e}}{\partial t}(t_0,z).
\]
Note that $h_{z_e}(t,gz)=h_{z_e}(t,z)$ for every $g\in G$ and every $(t,z)\in \mathbb{R}\times P$, i.e. $h_{z_e}(t,z)$ is $G$-invariant. Then, we can define a function $H_{z_e}:\mathbb{R}\times P_{\mu_e}\rightarrow \mathbb{R}$ of the form
\[
H_{z_e}(t,[z]):=h_{z_e}(t,z),\qquad \forall z\in {\bf J}^{-1}(\mu_e),\quad \forall t\in \mathbb{R},
\]
where $[z]$ stands for the equivalence class of $z\in {\bf J}^{-1}(\mu_e)$ in ${\bf J}^{-1}(\mu_e)/G_{\mu_e}$. Note that $H_{z_e}(t,[z])-k_{\mu_e}(t,[z])$ depends only on time. Hence, $H_{z_e}$ has an equilibrium point in $[z_e]$. 
Moreover,
\[
\frac{d}{dt}\bigg|_{t=t_0}H_{z_e}(t,[z(t)])=\frac{\partial h_{z_e} }{\partial t}(t_0,z),\qquad \forall t_0\in \mathbb{R},\qquad \forall [z]\in {\bf J}^{-1}(\mu_e)/ G_{\mu_e}.
\]
Let us use $H_{z_e}$ to study the stability of $[z_e]$ in $P_{\mu_e}$. In particular, we will study the conditions on $h$ to ensure that $H_{z_e}$ gives rise to different types of stable equilibrium points at $[z_e]$. With this aim, consider a coordinate system $\{x_1,\ldots,x_n\}$ on an open neighbourhood $U$ of $[z_e]\in P_{\mu_e}$ such that $x_i([z_e])=0$ for $i=1,\ldots, n$. Let $\alpha=(\alpha_1,\ldots,\alpha_n)$, with $\alpha_1,\ldots,\alpha_n\in \mathbb{N}\cup \{0\}$, be a multi-index with $n:=\dim {\bf J}^{-1}(\mu_e)/G_{\mu_e}$. Let $|\alpha|:=\sum_{i=1}^n\alpha_i$ and $D^\alpha:=\partial^{\alpha_1}_{x_1}\cdots \partial^{\alpha_n}_{x_n}$.

\begin{lemma}\label{Lemm:MegaLemma} Let us define the $t$-dependent parametric family of $n\times n$  matrices $M(t)$ with entries
\[
[M(t)]_i^j:=\frac12\frac{\partial^2H_{z_e}}{\partial x_i\partial x_j}(t,[z_e]),\qquad \forall t\in \mathbb{R},\qquad i,j=1,\ldots,n,
\]
and let ${\rm spec}(M(t))$ stand for the spectrum of  the matrix $M(t)$ at $t\in \mathbb{R}$. Assume that there exists a constant $\lambda$ such that $0<\lambda<\inf_{t\in I_{t^0}}\min {\rm spec}(M(t))$ for some $t^0\in\mathbb{R}$. Suppose also that there exists a real constant $c$ such that
\[
c\geq \frac 16\sup_{t\in I_{t^0}}\max_{|\alpha|=3}\max_{[y]\in \mathcal{B}}|D^\alpha H_{z_e}(t,[y])|
\]
for a certain compact neighbourhood $\mathcal{B}$ of $[z_e]$. Then, there exists an open neighbourhood $\mathcal{U} $  of $[z_e]$ where the function $H_{z_e}:\mathbb{R}\times \mathcal{U}\rightarrow \mathbb{R}$ is lpdf from $t^0$. If there exists additionally a constant $\Lambda$ such that $\sup_{t\in I_{t^0}}\max {\rm spec}(M(t))< \Lambda$, then $H_{z_e}:\mathbb{R}\times \mathcal{U}\rightarrow \mathbb{R}$ is a decrescent function from $t^0$. 
\end{lemma}
\begin{proof}
Since $z_e$ is a relative equilibrium point of $(M,\omega,h,{\bf J},\Phi)$, then $H_{z_e}(t,\cdot)$ has a critical point at $[z_e]$ for every $t\in \mathbb{R}$.
By the Taylor expansion of $H_{z_e}(t,\cdot)$ around $[z_e]$ and the fact that $z_e$ is a relative equilibrium point of each $h_{z_e}(t,\cdot)$ for $t\in \mathbb{R}$, one has that
\[
H_{z_e}(t,[z])=\frac{1}{2}\sum_{i,j=1}^n\frac{\partial^2H_{z_e}}{\partial x_i\partial x_j}(t,[z_e])x_ix_j+R_t([z]),\qquad [z]\in U,\quad t\in \mathbb{R}, 
\]
where $R_t([z])$ reads for the third-order remainder function for $H_{z_e}(t,[z])$ at a fixed $t\in \mathbb{R}$ around $[z_e]$. It is immediate that the coefficients of the quadratic part of the Taylor expansion match the matrix $M(t)$ in the coordinates $\{x_1,\ldots,x_n\}$. Since $M(t)$ is symmetric, it can be  diagonalised via an orthogonal transformation $O_t$ for each $t\in \mathbb{R}$. Let $\lambda_1(t),\ldots,\lambda_n(t)$ be the (possibly repeated) $n$ eigenvalues of $M(t)$ and let ${\bf w}=(w_1,\ldots,w_n)^T$ be the coordinate vector corresponding to ${\bf z}=(x_1,\ldots,x_n)^T$ in the diagonalising basis induced by $O_t$. In order words, ${\bf w}=O_t{\bf z}$. Although it is not necessary for our purposes, note that $O_t$ can be obtained by finding an orthogonal basis for each eigenvalue space for $M(t)$ and each time $t$. Then, ${\bf z}^TM(t){\bf z}={\bf w}^TD(t){\bf w}$, where $D(t)={\rm diag}(\lambda_1(t),\ldots,\lambda_n(t))$. Thus, ${\bf w}^TD(t){\bf w}=\sum_{i=1}^n\lambda_i(t)w_i^2$. Then,
\[
\frac{1}{2}\sum_{i,j=1}^n\frac{\partial^2H_{z_e}}{\partial x_i\partial x_j}(t,[z_e])x_ix_j= {\bf z}^TM(t){\bf z}={\bf w}^TD(t){\bf w}\geq \lambda (t)\|{\bf w}\|^2,
\]
where $\lambda(t):={\rm min}_{i=1,\ldots,n}\lambda_i(t)$ for each $t\in \mathbb{R}$. By our assumption on the existence of $\lambda>0$ and since $O_t$ is orthogonal, one gets that
\[
\frac{1}{2}\sum_{i,j=1}^n\frac{\partial^2H_{z_e}}{\partial x_i\partial x_j}(t,[z_e])x_ix_j\geq  \lambda (t)\|{\bf z}\|^2\geq  \lambda \|{\bf z}\|^2
\]
on $U$. Recall that the third-order Taylor remainder $R_t([z])$ around $[z_e]$ can be written as
\[
R_t([z])=\sum_{|\beta|=3}B_\beta(t,[z]){\bf z}^\beta,\qquad {\bf z}^\beta:=x_1^{\beta_1}\cdot\ldots\cdot x_n^{\beta_n},
\]
on points $[z]$ of the open coordinate subset $U$, $t\in \mathbb{R}$, and for certain functions $B_\beta:\mathbb{R}\times U\rightarrow \mathbb{R}$. Note that $R_t([z])$ is not a third-order polynomial due to the fact that the functions $B_\beta(t,[z])$ depend on the coordinates of $[z]$. Moreover, $R_t([z])$ can be bounded by a third-order polynomial in the coordinates of $[z]$ for each fixed time $t$ on an open $O_t$, but the open subsets $U_t$ and the coefficients of the polynomials used to bound $R_t([z])$ depend on $t$. This leads to problems since, for instance, to bound $R_t([z])$ for every $t\in \mathbb{R}$, one  will have to restrict to $\bigcap_{t\in I_{t^0}} U_t$, which may give rise to a single point.  Let us use then another, more appropriate but complicated, approach to bound all $R_t([z])$ for $t\in I_{t^0}$. 

The  $B_\beta$ are known to be bounded by 
\[
|B_\beta(t,[z])|\leq \frac{1}{3!}\max_{|\alpha|=3}\max_{y\in \mathcal{C}}|D^\alpha H_{z_e}(t,[y])|,\qquad \forall [z]\in \mathcal{C}
\]
on any compact neighbourhood $\mathcal{C}$ of $[z_e]$ for each $t\in\mathbb{R}$. By our assumptions, there exists a constant $c>0$ satisfying
\begin{equation}\label{Eq:Con0}
c\geq \frac{1}{3!}\max_{|\alpha|=3}\max_{y\in \mathcal{B}}|D^\alpha H_{z_e}(t,[y])|,\qquad \forall t\in I_{t^0},
\end{equation}
for some compact neighbourhood $\mathcal{B}$ of $[z_e]$. Let us prove that  
\[
\frac{1}{2}\sum_{i,j=1}^n\frac{\partial^2H_{z_e}}{\partial x_i\partial x_j}(t,[z_e])x_ix_j+R_t([z])-\frac 12\lambda\|{\bf z}\|^2
\]
is bigger or equal to zero for every $t\in  I_{t^0}$ and every $[z]\in \mathcal{U}$ for a certain open neighbourhood $\mathcal{U}$ of $[z_e]$. By our general assumptions, $\lambda< \inf_{t\in I_{t^0}} \lambda(t)$. Note that $\lambda_i(t)-\lambda\geq \lambda(t)-\lambda$ and $\lambda(t)-\lambda$ is larger than a certain properly chosen $\lambda'>0$ and every $t\in I_{t^0}$. Then,
\[
\frac{1}{2}\sum_{i,j=1}^n\frac{\partial^2H_{z_e}}{\partial x_i\partial x_j}(t,[z_e])x_ix_j-\lambda\|{\bf z}\|^2= {\bf w}^T{\rm diag}(\lambda_1(t)-\lambda,\ldots,\lambda_n(t)-\lambda){\bf w}\geq  \lambda'\|{\bf w}\|^2= \lambda'\|{\bf z}\|^2.
\]
Then, the first bracket in the following expression 
\[
\left(\frac{1}{2}\sum_{i,j=1}^n\frac{\partial^2H_{z_e}}{\partial x_i\partial x_j}(t,[z_e])x_ix_j- \lambda\|{\bf z}\|^2- \lambda'\|{\bf z}\|^2\right)+\left( \lambda'\|{\bf z}\|^2+R_t([z])\right).
\]
is larger or equal to zero on $U$. Let us prove the same for the second bracket on a neighbourhood of $[z_e]$. Note that
\[
|R_t([z])|\leq \sum_{|\beta|=3}|B_{\beta}(t,[z])||x_1|^{\beta_1}\cdot\ldots\cdot |x_n|^{\beta_n}\leq c\sum_{|\beta|=3}|x_1|^{\beta_1}\cdot\ldots\cdot|x_n|^{\beta_n},\qquad \forall t\in I_{t^0}
\]
on $\mathcal{B}$. The function
\[
 \lambda'\|{\bf z}\|^2-c\sum_{|\beta|=3}\lambda_\beta {\bf z}^\beta,
\]
where the $\{\lambda_\beta\}$ is any set of constants such that $\lambda_\beta\in \{\pm 1\}$ for every multi-index $\beta$ with $|\beta|=3$,
admits a minimum at $[z_e]$ as follows from standard differential calculus arguments. As a consequence, the above function is bigger or equal to zero on a neighbourhood $U_{\{\lambda_\beta\}}$ of zero. Considering the intersection of all the possible open subsets $U_{\{\lambda_\beta\}}$ for every set of constants $\lambda_\beta$, we obtain an open neighbourhood $\mathcal{U}$ of $[z_e]$. Assume that $[z]\in U$ is such that
\[
0> \lambda'\|{\bf z}\|^2-c\sum_{|\beta|=3}|x_1|^{\beta_1}\cdot \ldots\cdot |x_n|^{\beta_n} 
\]
Then,
\[
0> \lambda'\|{\bf z}\|^2-c\sum_{|\beta|=3}{\rm sgn}\left(\prod_{i=1}^nx^{\beta_i}_i\right) {\bf z}^\beta,
\]
where ${\rm sgn}(a)$ is the sign of the constant $a$. Then, $[z]$ cannot belong to $\mathcal{U}$. In other words, 
\begin{equation}\label{UpperBound}
 \lambda'\|{\bf z}\|^2-c\sum_{|\beta|=3}|x_1|^{\beta_1}\cdot \ldots\cdot |x_n|^{\beta_n} \geq 0
\end{equation}
on $\mathcal{U}$. Since $|R_t([z])|\leq c\sum_{|\beta|=3}|x_1|^{\beta_1}\cdot \ldots\cdot |x_n|^{\beta_n}$ on $\mathcal{U}$ and $t\in I_{t^0}$, then
\[
 \lambda'\|{\bf z}\|^2+R_t([z])\geq 0
\]
for every $[z]\in \mathcal{U}$ and $t\in I_{t^0}$. Finally, one gets that  
\[
H_{z_e}(t,[z])\geq \lambda \|{\bf z}\|^2,\qquad \forall [z]\in \mathcal{U},\qquad \forall t\in I_{t^0}. 
\]
Hence, the restriction of $H_{z_e}:\mathbb{R}\times P_{\mu_e}\rightarrow \mathbb{R}$ to $I_{t^0}\times \mathcal{U}$ is a lpdf function.

Now, the orthogonal change of variables $O_t$ allows us to write
\[
\frac{1}{2}\sum_{i,j=1}^n\frac{\partial^2H_{z_e}}{\partial x_i\partial x_j}(t,[z_e])x_ix_j={\bf z}^TM(t){\bf z}= {\bf w}^TD(t){\bf w}\leq   \Lambda (t)\|{\bf w}\|^2= \Lambda(t)\|{\bf z}\|^2,
\]
for $\Lambda(t):={\rm max}_{i=1,\ldots,n}\lambda_i(t)$ and every $t\in \mathbb{R}$. By assumption, $\Lambda> \Lambda(t)$ for every $t\in I_{t^0}$. Hence,
\[
\frac{1}{2}\sum_{i,j=1}^n\frac{\partial^2H_{z_e}}{\partial x_i\partial x_j}(t,[z_e])x_ix_j\leq  \Lambda\|{\bf z}\|^2,\qquad \forall t\in I_{t^0}.
\]
Recall the expression (\ref{UpperBound}) for every $t\in I_{t^0}$ and $[z]\in \mathcal{U}$.  Then, one has that
\[
H_{z_e}(t,[z])\leq  \Lambda \|{\bf z}\|^2+ \lambda'\|{\bf z}\|^2
\]
and $H_{z_e}$ is decrescent on $I_{t^0}\times \mathcal{U}$.
\end{proof}

It is worth noting that the eigenvalues of $M(t)$ depend on the chosen coordinate system around $[z_e]$. Choosing an appropriate coordinate system, one may simplify $M(t)$ at certain values of $t$ by writing $M(t)$ in a canonical form. Nevertheless, the simplification of $M(t)$ at every time $t\in I_{t^0}$ for a certain coordinate system around $[z_e]$ will be evidently impossible in most cases. One may still use $t$-dependent changes of variables to simplify $M(t)$ at every $t$ simultaneously, but finding such a $t$-dependent coordinate system may be difficult and it may be incompatible with the symplectic formalism, which concerns only time-independent changes of variables. We therefore restrict ourselves to determining a condition on a particular coordinate system.

By using the above lemma, we obtain the immediate theorem.

\begin{theorem}\label{Th:StabilityCon} Let assume that there exist $\lambda,c>0$ and an open neighbourhood $U$ of $[z_e]$ so that
\[
\lambda< {\rm min}({\rm spec}(M(t))),\qquad c\geq \frac{1}{3!}\max_{ |\alpha|=3}\sup_{[x]\in U}|D^\alpha H_{z_e}(t,[x])|,\qquad
\frac{\partial H_{z_e}}{\partial t}\bigg|_{U}\leq 0,
\]
for every $t\in I_{t^0}$, then $[z_e]$ is a stable point of the Hamiltonian system $k_{\mu_e}$ on ${\bf J}^{-1}(\mu_e)/G_{\mu_e}$ from $t^0$. If there exists $\Lambda$ such that $\max ({\rm spec}(M(t)))<\Lambda$ for every $t\in I_{t^0}$, then $[z_e]$ is uniformly  stable from $t^0$.
\end{theorem}
\begin{proof}  
By Lemma \ref{Lemm:MegaLemma} and our given assumptions, $H_{z_e}(t,[z])$ is a locally positive definite $C^1$-function. By Theorem \ref{Th:BasicTheoremOfLyapunov} (point 1.) and  $\partial H_{z_e}/\partial t\leq 0$, we obtain that $[z_e]$ is stable from $t^0$. If additionally $\Lambda$ exists, then again Theorem \ref{Th:BasicTheoremOfLyapunov} (point 2.) shows that $[z_e]$ is uniformly stable from $t^0$. 
\end{proof}

For geometrical reasons, let us consider the following corollary, which follows from Theorem \ref{Th:StabilityCon} by assuming an stronger condition on the derivatives of $H_{z_r}$. More specifically,  Corollary \ref{Cor::StablePoint} provides conditions that can be proved to hold independently of the chosen coordinate system.

\begin{corollary}\label{Cor::StablePoint}
If there exist $\lambda,c>0$ and an open neighbourhood $U$ of $[z_e]$ such that
\[
\lambda <\min\left({\rm spec}\left(M(t)\right)\right),\qquad c\geq \frac{1}{3!}\max_{1\leq |\alpha|\leq 3}\sup_{[x]\in U}\left|D^\alpha H_{z_e}(t,[x])\right|,\qquad \frac{\partial H_{z_e}}{\partial t}\bigg|_{U}\leq 0,
\]
for every $t\in I_{t^0}$, then $[z_e]$ is a uniformly  stable point of the Hamiltonian system $k_{\mu_e}$ on $\mathbf{J}^{-1}(\mu_e)/G_{\mu_e}$ from $t^0$.
\end{corollary}

Note that the existence of $c$ in Corollary \ref{Cor::StablePoint} implies that $\max({\rm spec}(M(t)))$ for every $t\in I_{t^0}$ is bounded from above:

\[
v^TM(t)v\leq \sum_{i,j=1}^n|v_iv_jM^j_i(t)|\leq \sum_{i,j=1}^n|v_i||v_j||M^i_j(t)|\leq 6c \sum_{i,j=1}^n\|v\|^2=6cn^2 \|v\|^2,\qquad \forall v\in \mathbb{R}^n.
\]
Hence, $v^TM(t)v<\Lambda v^Tv$ for $v\in \mathbb{R}^n\backslash\{0\}$ and $\Lambda> 6cn^2$.

Note that previous results use a distance on an open coordinate neighbourhood of $[z_e]$ induced by a standard norm in $\mathbb{R}^n$. As the topology induced by this norm on the open neighbourhood of $[z_e]$ is the same as the one induced by any other one Riemannian metric,  our results concerning the stability of $[z_e]$ are independent of the used metric.

Let us now prove that the conditions contained in Corollary \ref{Cor::StablePoint} have a geometric meaning: if they hold in a coordinate system on a neighbourhood of $[z_e]$, then they hold in any other such a coordinate system for, eventually, other values of the constants $\lambda,c$. Moreover,  our comments partially apply to Theorem \ref{Th:StabilityCon}. First, the condition in the time-derivative of $H_{z_e}$ is defined in a coordinate-free manner relative to coordinates on $P_{\mu_e}$. Meanwhile, the remaining conditions in Corollary \ref{Cor::StablePoint} require a more detailed analysis.  
\begin{lemma}\label{Lem:Coorl} If the $t$-dependent matrix $M(t)$, which is defined in a local coordinate system $\{x_1,\ldots,x_n\}$ on an open neighbourhood of an equilibrium point $[z_e]\in P_{\mu_e}$, satisfies that $0< \lambda<{\rm inf}_{t\in I_{t^0}}\min{\rm spec}\,M(t) $  for some $\lambda$ (resp. ${\rm sup}_{t\in I_{t^0}}\max {\rm spec}\,M(t)< \Lambda$ for some $\Lambda$), then $M_{\mathcal{B}'}(t)$, defined as $M(t)$ but in  another coordinate system $\mathcal{B}':=\{\tilde{x}_1,\ldots,\tilde{x}_n\}$ on another neighbourhood in $P_{\mu_e}$ of $[z_e]$, satisfies that $0< \lambda'<{\rm inf}_{t\in I_{t^0}}\min {\rm spec}\,M_{\mathcal{B}'}(t)$  for some $\lambda'$  (resp. ${\rm sup}_{t\in I_{t^0}}\max{\rm spec}\,M_{\mathcal{B}'}(t)<\Lambda'$ for some $\Lambda'$).  
\end{lemma}
\begin{proof} Since every symmetric metric can be diagonalised orthogonally into a diagonal matrix, $D(t)$, via a $t$-dependent orthogonal matrix $O_t$, the condition for $M(t)$  amounts to the fact that 
\begin{equation}\label{con}
v^TM(t)v=v^TO^T_tD(t)O_tv> \lambda v^TO^T_tO_tv= \lambda v^Tv,\qquad \forall v\in \mathbb{R}^n\backslash\{0\},\qquad \forall t\in I_{t^0}.
\end{equation}
Let us consider $M_{\mathcal{B}'}(t)$. Since $[z_e]$ is an equilibrium point of $H_{z_e}$, there exists an invertible $n\times n$  time-independent matrix $A$ such that
\[
M_{\mathcal{B}'}(t)=A^TM(t)A,\qquad \forall t\in \mathbb{R}.
\]
Then,
\[
v^TM_{\mathcal{B}'}(t)v=(Av)^TM(t)Av>  \lambda(Av)^TAv,\qquad \forall v\in \mathbb{R}^n\backslash\{0\},\qquad  \forall t\in I_{t^0}.
\]

Since $A$ is invertible, the positive function $f:v\in S^{n-1}\mapsto  (Av)^T(Av)\in \mathbb{R}$ on the ball $S^{n-1}=\{v\in \mathbb{R}^n:||v||:=\sqrt{v^Tv}=1\}$, which is compact, reaches a maximum and a minimum $M_S,m_S>0$, respectively. Then, $(Av)^T(Av)\geq  m_S v^Tv$ for every $v\in \mathbb{R}^n$. Thus,
\[
v^TM_{\mathcal{B}'}(t)v> \lambda m_Sv^Tv,\qquad \forall v\in \mathbb{R}^n\backslash\{0\},\qquad  \forall t\in I_{t^0}.
\]
Similarly, we have that $(Av)^T(Av)\leq M_Sv^Tv$ for every $v\in\mathbb{R}^n$. Then, the existence of $\Lambda$ leads to
\[
v^TM_{\mathcal{B}'}(t)v<\Lambda M_Sv^T v,\qquad \forall{v}\in \mathbb{R}^n\backslash\{0\}, \qquad\forall t\in I_{t^0}.
\]
Choosing $\lambda'=\lambda m_S$ and $\Lambda'=\Lambda M_S$, the lemma follows.
\end{proof}

Note the that the condition for $c$ in Corollary \ref{Cor::StablePoint} is independent of the chosen coordinate system. More specifically, the condition on a new coordinate system also holds by choosing a new $c'$ and restricting to a new open subset of $[z_e]$ where the previous condition and the new coordinate systems are defined.

\section{Stability, reduced space, and relative equilibrium points}

The main idea of the energy-momentum method is to determine some properties of $h$ on a neighbourhood of a relative equilibrium point $z_e$ in $P$ that ensure a certain type of stability at an equilibrium point of $k_{\mu_e}$  in $P_{\mu_e}$. In particular, we hereafter   give conditions on the functions $h^t_{\mu_e}:z\in   {\bf J}^{-1}(\mu_e)\mapsto h(t,z)\in \mathbb{R}$, and $\partial h^t_{\mu_e}/\partial t$ with $t\in I_{t^0}$, to ensure that the conditions in Theorem \ref{Th:StabilityCon} and/or Corollary \ref{Cor::StablePoint} hold. Instead of inspecting  $M(t)$,  we will search for conditions on the functions $h_{\xi,t}$ for $t\in I_{t^0}$, which is more practical as the latter are not defined on the quotient of a submanifold of $P$ and therefore are available without making additional computations. Note that the ideas used to prove Proposition \ref{Stwi} and Corollary \ref{2varvanish} below are a  generalisation of the $t$-independent formulation of the energy-momentum method in \cite{MS88}. Before we proceed to the following proposition, let us define $(\delta^2f)(X,Y):=\iota_Yd(\iota_Xdf)$ for every $X,Y\in \mathfrak{X}(P)$ and $f\in C^\infty(P)$. If $f$ is such that $df_p=0$ for a certain $p\in P$, then  $[\delta^2f(X,Y)](p)$ depends only on the values of $X,Y$ at $p$, which gives rise to a bilinear map on $T_pP$ of the form
\[
(\delta^2f)_p(v,w):=(\iota_Yd(\iota_Xdf))(p),\qquad \forall v,w\in T_pP,
\]
for $X,Y\in \mathfrak{X}(P)$ such that $X(p)=v$ and $Y(p)=w$. Moreover, $(\delta^2f)_p$ becomes then symmetric.

\begin{proposition}\label{Stwi}
Let $z_{e}\in P$ be a relative equilibrium point for $(P,\omega,h,\Phi,\mathbf{J})$. Then, 
\begin{equation}\label{2var}
(\delta^{2}h_{\xi,t})_{z_{e}}((\eta_{P})_{z_{e}}, v_{z_e})=0, \quad\forall\eta\in \mathfrak{g},\quad \forall v_{z_e}\in T_{z_e}\mathbf{J}^{-1}(\mu_e),\quad\forall t\in \mathbb{R}.
\end{equation}
\end{proposition}
\begin{proof}
The $G$-invariance of $h:\mathbb{R}\times P\rightarrow \mathbb{R}$ and  the equivariance condition for ${\bf J}$ yields 
\[
h_{\xi,t}(gp)=h(t,gp)-\langle\mathbf{J}(gp), \xi(t)\rangle + \langle\mu_{e}, \xi(t)\rangle=h(t,p)-\langle {\rm  Ad}_{g^{-1}}^{*}(\mathbf{J}(p)),\xi(t)\rangle+\langle\mu_{e}, \xi(t)\rangle\]
and  
\[h_{\xi,t}(gp)=h(t,p)-\langle \mathbf{J}(p),{\rm Ad}_{g^{-1}}(\xi(t))\rangle+\langle \mu_{e}, \xi(t)\rangle,
\]
for any $g\in G$ and $p\in P$. Substituting $g:=\exp(s\eta)$, with $\eta\in \mathfrak{g}$, and differentiating with respect to the parameter $s$, one obtains 
\[
(\iota_{\eta_P}d h_{\xi,t})(p)=-\left\langle \mathbf{J}(p),\frac{d}{ds}\bigg|_{s=0}{\rm Ad}_{\exp(-s\eta)}(\xi(t))\right\rangle=\langle \mathbf{J}(p),[\eta,\xi(t)]\rangle.
\]
Taking variations relative to $p\in P$ above, evaluating at $z_{e}$, and since $(d h_{\xi,t})_{z_{e}}=0$ because $z_e$ is a critical point of $h_{\xi,t}$, one has that \[(\delta^{2}h_{\xi,t})_{z_{e}}((\eta_{P})_{z_{e}},v_{z_e})=\langle T_{z_{e}}{\bf J}(v_{z_e}), [\eta, \xi(t)]\rangle,
\]
which vanishes if $T_{z_{e}}{\bf J}(v_{z_e})=0$, i.e. if $v_{z_e}\in \mathrm{ker}[T_{z_{e}}{\bf J}]=T_{z_{e}}{\bf J}^{-1}(\mu_{e})$.
\end{proof}
Propositions \ref{Stwi} and  \ref{2.2} yield the following.
\begin{corollary}\label{2varvanish}
   The mapping $(\delta^{2}h_{\xi,t})_{z_{e}}$ vanishes identically on $T_{z_e}(G_{\mu_e}z_e)$ for every $t\in \mathbb{R}$.
\end{corollary}

\begin{proof}
Proposition \ref{2.2} shows that $T_{z_{e}}(G_{\mu_{e}}  z_{e})=T_{z_{e}}(G  z_{e})\cap\mathrm{ker}[T_{z_{e}}{\bf J} ]$. Since $T_{z_{e}}(G_{\mu_{e}} z_{e})\subset T_{z_{e}}(G z_{e})$, the result follows from (\ref{2var}) by taking $v_{z_e}:=(\xi_{P})_{z_{e}}$, with $\xi\in \mathfrak{g}_{\mu_{e}}$.
\end{proof}

Recall that we assume that $G_{\mu_e}$ acts freely and properly on ${\bf J}^{-1}(\mu_e)$. Consider a set of coordinates $\{z_1,\ldots,z_q\}$ on an open $\mathcal{A}\subset {\bf J}^{-1}(\mu_e)$ containing $z_e$.  Let $\{\pi^*_{\mu_e}x_1,\ldots,\pi^*_{\mu_e}x_n\}$ be the coordinates on $\mathcal{A}$ given by the pullback to $\mathcal{A}$ of certain coordinates $\{x_1\ldots,x_n\}$ on $\mathcal{O}:=\pi_{\mu_e} (\mathcal{A})$ \footnote{To simplify the notation, we will write $\{x_1,\ldots,x_n\}$ for a set of coordinates on a certain neighbourhood of $[z_e]$ and their pull-backs to $\mathbf{J}^{-1}(\mu_e)$ via $\pi_{\mu_e}$.} and let $\{y_{1},\ldots,y_{s}\}$ be additional coordinates giving rise to a coordinate system $\{z_1,\ldots,z_q\}$ on $\mathcal{A}$. Due to the $G_{\mu_e}$-invariance of $h_{\mu_e}:=h\circ \iota_{\mu_e}:{\bf J}^{-1}(\mu_e)\rightarrow \mathbb{R}$, one has that there exists $c$ such that 
\[
c\geq \frac{1}{3!}\max_{3= |\vartheta|}\sup_{z\in \mathcal{A}}|D^\vartheta h_{\mu_e}(t,y)|,\qquad \forall t\in I_{t^0},
\]
where $\vartheta$ is a multi-index $\vartheta:=(\vartheta_1,\ldots,\vartheta_q)$, 
if and only if 
\begin{equation}\label{Rel2}
c\geq \frac{1}{3!}\max_{3= |\alpha|}\sup_{x\in \mathcal{O}}|D^\alpha H_{z_e}(t,x)|,\qquad \forall t\in I_{t^0},
\end{equation}
for $\mathcal{O}$, which is an open neighbourhood of $[z_e]$ because $\pi_{\mu_e}$ is an open mapping. Indeed, since $h_{\mu_e}$ is constant on the submanifolds where $x_1,\ldots,x_n$ take constant values,  $h_{\mu_e}(t,x_1,\ldots,x_n,y_1,\ldots,y_s)-h(t,z_e)=H_{z_e}(t,x_1,\ldots,x_n)$ and (\ref{Rel2}) follows. 

Consider again the coordinate system $\{z_1,\ldots,z_q\}$ on ${\bf J}^{-1}(\mu_e)$. We write $[\widehat{M}(t)]$ for the  $t$-dependent $q\times q$ matrix 
\[
[\widehat{M}(t)]_i^j:=\frac{\partial^2h_{\mu_e}}{\partial z_i\partial z_j}(t,z_e),\qquad i,j=1,\ldots,q.
\]

Lemma \ref{Lem:Coorl} tells us that,  geometrically, the existence of $\lambda$ and $\Lambda$ amounts to the fact that the $t$-dependent bilinear symmetric form $K(t):T_{[z_e]}P_{\mu_e}\times T_{[z_e]}P_{\mu_e}\rightarrow \mathbb{R}$ given by
\[
K(t)=\frac 12\sum_{i,j=1}^n\frac{\partial^2H_{z_e}}{\partial x_i\partial x_j}(t,[z_e])dx_i|_{[z_e]}\otimes dx_j |_{[z_e]}
\]
satisfies that
\begin{equation}\label{Eq:Comen}
K(t)(w,w)> \lambda (w|w)_{\mathcal{B}},\qquad \forall w\in T_{[z_e]}P_{\mu_e}\backslash\{0\},\quad\forall t\in I_{t^0},
\end{equation}
where $(\cdot |\cdot )_{\mathcal{B}}$ is the Euclidean product in $T_{[z_e]}P_{\mu_e}$ satisfying that $\{\partial_{x_1},\ldots,\partial _{x_n}\}$ is an orthonormal basis. In fact,  if $v$ is the column vector describing the coordinates of $w\in T_{[z_e]}P_{\mu_e}$ in the chosen orthonormal basis, then  (\ref{Eq:Comen}) can be rewritten as
\[
K(t)(w,w)=v^TM(t)v> \lambda v^Tv=\lambda (w|w)_{\mathcal{B}}, \qquad \forall w\in T_{[z_e]}P_{\mu_e}\backslash\{0\},\quad \forall t\in I_{t^0}.
\]
Note that, for any other inner product $(\cdot| \cdot)_{\mathcal{B}'}$ on $T_{[z_e]}P_{\mu_e}$, there exists $m_i,m_s>0$ such that $m_s(w|w)_{\mathcal{B}'}\geq (w|w)_{\mathcal{B}}\geq m_i(w|w)_{\mathcal{B}'}$ for all $w\in T_{[z_e]}P_{\mu_e}$. Hence,  if condition (\ref{Eq:Comen}) holds for an inner product in $T_{[z_e]}P_{\mu_e}$, it is also satisfied for any other inner product in $T_{[z_e]}P_{\mu_e}$ with another positive $\lambda$. A similar reasoning can be applied to the relation $\Lambda  (w|w)_\mathcal{B}> K(t)(w,w)$ for some $\Lambda>0$, for all $t\in I_{t^0}$ and every $w\in T_{[z_e]}P_{\mu_e}\backslash\{0\}$.

The reason to introduce the inner product $(\cdot|\cdot)_\mathcal{B}$ is theoretical and practical. To effectively determine whether the $t$-dependent matrix $M(t)$ has eigenvalues that can be bounded from below simultaneously for every time $t\in I_{t^0}$, we plan to use the eigenvalues of the matrix representation of $K(t)$ and $(\cdot|\cdot)_\mathcal{B}$, which are geometric objects. The fact that $(\cdot|\cdot)_{\mathcal{B}}$ may be chosen arbitrary simplifies to verify the condition.

Let us show how condition (\ref{Eq:Comen}) can be checked via an object defined on the space $\mathbf{J}^{-1}(\mu_e)$. Since $h_{\mu_e}$ has a critical point at each relative equilibrium point $z_e\in {\bf J}^{-1}(\mu_e)$, there exists a $t$-dependent bilinear symmetric function  $\widehat{M}(t):T_{z_e}{\bf J}^{-1}(\mu_e)\times T_{z_e}{\bf J}^{-1}(\mu_e)\rightarrow \mathbb{R}$ of the form
\[
\widehat{M}(t):=\frac 12 \sum_{i,j=1}^q\frac{\partial^2h_{\mu_e}}{\partial z_i\partial z_j}(t,z_e)dz_i|_{z_e}\otimes dz_j|_{z_e},\qquad \forall t\in I_{t^0},
\]
where $\mathcal{B}=\{z_1,\ldots,z_q\}$ is any coordinate system in an open neighbourhood of $z_e\in {\bf J}^{-1}(\mu_e)$.

Let us consider the coordinate system $\{x_1,\ldots,x_n,y_1,\ldots,y_s\}$ on the open neighbourhood $z_e$ in $\mathbf{J}^{-1}(\mu_e)$ defined above. In this coordinate system, we obtain
\[
\frac{\partial^2h_{\mu_e}}{\partial x_k\partial y_j}(t,z_e)=\frac{\partial^2h_{\mu_e}}{\partial y_i\partial y_j}(t,z_e)=0, \qquad i,j=1,\ldots,s,\qquad k=1,\ldots,n,\qquad \forall t\in \mathbb{R}.
\]
In the chosen coordinate system, one sees that $\pi_{\mu_e}^*K(t)=\widehat{M}(t)$ and $T_{z_e}(G_{\mu_e}z_e)\subset\ker \widehat{M}(t)$ for every $t\in \mathbb{R}$. The latter relation holds in any other coordinate system. Hence,  $K(t)$ can be considered as the induced bilinear form by $\widehat{M}(t)$ on  $S_{z_e}:=T_{z_e}{\bf J}^{-1}(\mu_e)/T_{z_e}(G_{\mu_e}z_e)\simeq T_{[z_e]}P_{\mu_e}$. Thus,  the conditions for $M(t)$  can be tested straightforwardly via an object in ${\bf J}^{-1}(\mu_e)$, namely $\widehat{M}(t)$. Note also that if $\ker \widehat M(t)$ has dimension bigger than $\dim T_{z_e}(G_{\mu_e}z_e)$, the conditions of Lemma \ref{Lemm:MegaLemma} do not hold.
Corollary \ref{Cor::StablePoint} and the previous remarks give rise to the following theorem.

\begin{theorem}\label{Th:StabilityCon2} Let us assume that there exist $\lambda,c>0$ and an open coordinate neighbourhood $\mathcal{A}\subset {\bf J}^{-1}(\mu_e)$ of $z_e$ so that
\begin{equation}\label{eq:Cond}
\lambda< {\rm min}({\rm spec}([\widehat{M}(t)]|_{S_{z_e}}),\quad c\geq \frac{1}{3!}\max_{1\leq |\vartheta|\leq 3}\sup_{y\in \mathcal{A}}|D^\vartheta h_{\mu_e}(t,y)|,\quad
\frac{\partial h_{\mu_e}}{\partial t}\bigg|_{\mathcal{A}}\leq 0,
\end{equation}
for every $t\in I_{t^0}$, then $[z_e]$ is a uniformly  stable point of the Hamiltonian system $k_{\mu_e}$ on ${\bf J}^{-1}(\mu_e)/G_{\mu_e}$ from $t^0$.

\end{theorem}

Recall that in the case of an autonomous Hamiltonian, the third condition in (\ref{eq:Cond}) is immediately satisfied. Moreover, still in the case of autonomous systems, if $h$ is smooth enough, there always exists the required $c$ for a certain open neighbourhood $\mathcal{A}$ of $z_e$. Finally, the condition on $\lambda$ boils down to the standard condition on the positiveness of the eigenvalues of the matrix $\widehat{M}$, which is not time-independent by assumption, up to the subspaces where it always vanishes due to Corollary \ref{2varvanish} (cf. \cite{MS88}).

Note that, in the non-autonomous case, the second condition in (\ref{eq:Cond}) can easily be verified for smooth enough functions $h$ whose spatial partial derivatives do not grow indefinitely in time. In fact, this is a condition rather easy to be satisfied.

Finally, let us relate the properties of $h_{\xi,t}$ with $H_{\mu_e}$ so as to study relative equilibrium points and their associated equilibrium points in $P_{\mu_e}$. Since $h_{\xi,t}$ has a critical point at a relative equilibrium point $z_e\in P$ for every $t\in \mathbb{R}$, we can define the $t$-dependent bilinear symmetric form on $T_{z_e}P$ given by
\[
T_{z_e}(t):=\frac 12\sum_{i,j=1}^{\chi}\frac{\partial^2 h_{\xi,t}}{\partial u_i\partial u_j}(t,z_e)du_i|_{z_e}\otimes du_j|_{z_e},\qquad \forall t\in \mathbb{R},
\]
where $u_1,\ldots,u_\chi$, with $\chi=\dim P$, is a coordinate system on an open neighbourhood of $z_e$ in $P$. Our aim now is to determine the relation of $T_{z_e}(t)$ with the matrix $\widehat{M}(t)$ to study the latter by means of the former. It is worth stressing that $T_{z_e}(t)$ is a geometric object easy to be constructed as it is defined on $T_{z_e}P$ and it depends essentially only on $h$ and ${\bf J}$.

Since ${\bf J}$ is regular, its coordinates, let us say $\mu_1,\ldots,\mu_r$, give rise to $\dim \mathfrak{g}$ functionally independent functions on $P$. Consider now the coordinate system on a neighbourhood of $z_e$ in ${\bf J}^{-1}(\mu_e)$ given by $x_1,\ldots,x_n,y_1,\ldots,y_s$.  Such coordinate functions on ${\bf J}^{-1}(\mu_e)$ can be extended to an open neighbourhood in $P$, containing the point $z_e$, smoothly. As ${\bf J}$ is regular at $z_e$, the functions $\mu_1,\ldots,\mu_r$, which are constant on the leaves of ${\bf J}$, satisfy $d\mu_1\wedge \ldots\wedge d\mu_r\neq 0$ on $z_e$. Thus, we obtain a coordinate system $x_1,\ldots,x_n,y_1,\ldots,y_s,\mu_1,\ldots,\mu_r$ on an open neighbourhood in $P$ containing $z_e$. Taking this into account, one obtains that
\[
\frac{\partial h_{t}}{\partial y_i}\bigg|_{{\bf J}^{-1}(\mu_e)}\!\!\!\!\!\!\!=0,\quad\frac{\partial \langle \mathbf{J}-\mu_{e},\ \xi(t)\rangle}{\partial y_i}=0,\qquad \forall t\in \mathbb{R},\quad i=1,\ldots,s.
\]
It is relevant to recall that the derivative $\partial h_{t}/\partial y_i$, with $i=1,\ldots,s$ do not need to vanish away from ${\bf J}^{-1}(\mu_e)$ because $y_1,\ldots,y_s$ were defined just as a smooth extension away from ${\bf J}^{-1}(\mu_e)$ without demanding any special property off ${\bf J}^{-1}(\mu_e)$. Moreover, 
\[
\left(\frac{\partial}{\partial y_j} \frac{\partial h_{t}}{\partial y_i}\right)\bigg|_{{\bf J}^{-1}(\mu_e)}\!\!\!\!\!\!\!\!\!\!\!=0,\quad \qquad \left(\frac{\partial}{\partial x_k} \frac{\partial h_{t}}{\partial y_i}\right)\bigg|_{{\bf J}^{-1}(\mu_e)}\!\!\!\!\!\!\!\!\!\!\!=0,
\]
\[
\frac{\partial}{\partial y_j} \frac{\partial \langle \mathbf{J}-\mu_{e},\ \xi(t)\rangle}{\partial y_i}=0,\qquad \frac{\partial}{\partial x_k} \frac{\partial \langle \mathbf{J}-\mu_{e},\ \xi(t)\rangle}{\partial y_i}=0,
\]
for all $t\in  \mathbb{R}, i,j=1,\ldots,s, k=1,\ldots,n$. 
Note that the first and second relations above hold because the derivative on the left depends, on points of ${\bf J}^{-1}(z_e)$, only on the values of  $\partial h_{t}/\partial y_i$ within ${\bf J}^{-1}(\mu_e)$. Nevertheless, this shows that, in the chosen coordinate system, the Hessian matrix of $h_{\xi,t}$, let us say ${\bf H}h_{\xi,t}$, on $T_{z_e}{\bf J}^{-1}(\mu_e)$ coincides with $\widehat M(t)$ on the chosen coordinate system. This is the main point: we can use $h_{\xi,t}$ to study $\widehat{M}(t)$ and $M(t)$. Since $h$ has not, in general, a critical point in $z_e$, the Hessian of $h$ at $z_e$ does not give rise to a bilinear symmetric form at $z_e$ but, at the chosen coordinate system, matches the matrix of $T_{z_e}(t)$.

 \section{ Example: The almost-rigid body}

Let us illustrate our $t$-dependent energy-momentum method via a generalisation of the standard example of the freely spinning rigid body \cite{MS88}. Our aim is to determine its relative equilibrium points and to study the second-order variation of the extended Hamiltonian $h_{\xi,t}$ and to generalise the autonomous result obtained in \cite{MS88}. Our main results are given in (\ref{eqcondex}) and (\ref{Eq:SecondVar}).

Let us consider $t^0=0$ and  let $SO_3$ be the Lie group of all orthogonal unimodular linear automorphisms on the Euclidean space $\mathbb{R}^3$. The Lie algebra of $SO_3$, let us say $\mathfrak{so}_3$, consists of all the $3\times 3$ skew-matrices and it can be identified with $\mathbb{R}^{3}$ via the standard isomorphism 
{\small \begin{equation}
\label{ISOex2}
\phi: \mathbb{R}^3\rightarrow \mathfrak{so}_3, \,\, \Omega \mapsto \widehat{\Omega}:=\left[\begin{array}{ccc}0&-\Omega^3&\Omega^2\\\Omega^3&0&-\Omega^1\\-\Omega^2&\Omega^1&0
\end{array}\right],
\end{equation}}where $\Omega:=(\Omega^1,\Omega^2,\Omega^3)^T$. Let `$\times$' be the vector product in $\mathbb{R}^3$. Then, $\widehat{\Omega}{\bf r}=\Omega\times {\bf r},\,\, [\widehat{\Omega},\widehat{\Theta}]=\widehat{\Omega\times \Theta}$, and $\Lambda\widehat{\Theta}\Lambda^T=\widehat{\Lambda\Theta}$ for every $\Lambda\in SO_3$, and every $\Theta,\Omega\in \mathbb{R}^3$. Hence, $\phi$ is a Lie algebra isomorphism between $\mathbb{R}^3$ (which is a Lie algebra relative to the vector product) and $\mathfrak{so}_3$ with the commutator of matrices. 

The {\it adjoint action} ${\rm Ad}:SO_3\times\mathfrak{so}_3 \rightarrow \mathfrak{so}_3$, defined geometrically in (\ref{adaction}), reduces to the expression ${\rm Ad}_\Lambda\widehat{\Theta}=\Lambda \widehat{\Theta}\Lambda^T$, as $\Lambda^{-1}\!=\!\Lambda^T$, for all $\Lambda \in SO_3$ and $\Theta\in \mathbb{R}^3$. Moreover, 
\[
\widehat{\Lambda({\bf r}\!\!\times\!\!{\bf s})}=\Lambda \widehat{{\bf r}\!\!\times\!\!{\bf s}}\Lambda^T=\Lambda[\widehat{\bf r},\widehat{\bf s}]\Lambda^T =[\Lambda \widehat{\bf r}\Lambda^T,\Lambda \widehat{\bf s}\Lambda^T]=[\widehat{\Lambda {\bf r}},\widehat{\Lambda {\bf s}}]=\widehat{\Lambda {\bf r}\!\!\times\!\! \Lambda {\bf s}}, \qquad \forall {\bf r},{\bf s}\in \mathbb{R}^{3}.
\]

One can identify $T_{\Lambda}SO_3$ with $\mathfrak{so}_3$ via two isomorphisms. Recall that $L_\Lambda: \Theta\in SO_3\mapsto \Lambda \Theta\in SO_3$ and $R_\Lambda:\Theta\in SO_3\mapsto \Theta \Lambda \in SO_3$ are diffeomorphisms for every $\Lambda\in SO_3$. Then, $T_{\rm Id_3}L_\Lambda:T_{\rm Id_3}SO_3\simeq \mathfrak{so}_3\mapsto T_\Lambda SO_3$ and $T_{\rm Id_3}R_\Lambda:T_{\rm Id_3}SO_3\simeq \mathfrak{so}_3\mapsto T_\Lambda SO_3$, where ${\rm Id}_3$ is the $3\times 3$ identity matrix, are  isomorphisms. We define $(T_{\rm Id_3}L_\Lambda)\widehat 
\Theta=:(\Lambda,\Lambda \widehat{\Theta})$, for every $\Theta\in \mathbb{R}^3$. Then, $(\Lambda,\Lambda \widehat{\Theta})$ is called  the {\it  left-invariant extension} of $\widehat{\Theta}$.  Meanwhile, we set $(T_{{\rm Id}_3}R_\Lambda)\widehat
\theta:=(\Lambda,\widehat{\theta}\Lambda)$, for every $\theta\in\mathbb{R}^3$. It is said that
$(\Lambda,\widehat{\theta}\Lambda)$ is the {\it  right-invariant extension} of $\widehat{\theta}$. We omit the base point, if it is known from context. We write $\Lambda \widehat{\Theta}$ and $\widehat{\theta}\Lambda$ for $(\Lambda,\Lambda \widehat{\Theta})$ and $(\Lambda,\widehat{\theta}\Lambda)$, respectively.

Since $\mathfrak{so}_3$ is a simple Lie algebra, its Killing metric, $\kappa$, is non-degenerate, which gives an isomorphism
\begin{equation}\label{isoex}
\widehat{\Theta}\in \mathfrak{so}_3\mapsto \kappa(\widehat{\Theta},\cdot)\in \mathfrak{so}^*_3.
\end{equation}
In particular, $\kappa$ reads, up to a non-zero optional proportional constant, as $\kappa(\widehat{\Theta},\widehat{\Omega})\!\!=\!\!\frac 12{\rm tr}(\widehat{\Theta}^T\widehat{\Omega})$, for all $\Theta,\Omega\! \in\! \mathbb{R}^3$. Moreover, $\Pi\cdot \Upsilon=\kappa(\widehat{\Pi},\widehat{\Upsilon})$, for all $\Pi, \Upsilon\in \mathbb{R}^3$ and the canonical Euclidean product "$\cdot$" in $\mathbb{R}^3$. This extends to
\[
\langle\Lambda \widehat{\Pi},\Lambda \widehat{\Theta}\rangle:=\frac{1}{2} {\rm tr} ((\Lambda \widehat{\Pi})^{T}\Lambda\widehat{\Theta})=\frac{1}{2}{\rm tr}(\widehat{\Pi}^{T}\widehat{\Theta})=\Pi\cdot \Theta,\quad \forall \Theta,\,\Pi\in\mathbb{R}^3.
\]
Moreover,
\[
\langle \widehat{\Pi}\Lambda,\widehat{\Theta}\Lambda \rangle:=\frac{1}{2} {\rm tr} (( \widehat{\Pi}\Lambda)^{T}\widehat{\Theta}\Lambda)=\frac{1}{2}{\rm tr}(\widehat{\Pi}^{T}\widehat{\Theta})=\Pi\cdot \Theta,\quad \forall \Theta,\,\Pi\in\mathbb{R}^3.
\]

For simplicity, $\widehat{\Pi}\in \mathfrak{so}_3^*$  will represent $\kappa(\widehat{\Pi},\cdot)\in \mathfrak{so}^*_3$ and elements of $T_{\Lambda}^{*}SO_3$ will by written as $(\Lambda,\widehat{\pi}\Lambda )$ and $(\Lambda,\Lambda\widehat{\Pi}).$ 
If $(\Lambda,\widehat{\pi}\Lambda )=(\Lambda,\Lambda\widehat{\Pi})$, then  $\widehat{\pi}=\Lambda\widehat{\Pi}\Lambda^T$, which matches the coadjoint action. Indeed,
\begin{multline*}
\langle {\rm Ad}_{\Lambda^T}^*\widehat{\Pi},\cdot\rangle=\frac 12{\rm Tr}(\widehat{\Pi}^T{\rm Ad}_{\Lambda^T}(\cdot))=\frac 12{\rm Tr}(\widehat{\Pi}^T\Lambda^T(\cdot)\Lambda)\\=\frac 12{\rm Tr}(\Lambda\widehat{\Pi}^T\Lambda^T(\cdot))=\frac 12{\rm Tr}((\Lambda \widehat{\Pi}\Lambda^T)^T(\cdot))=\langle \widehat{\pi},\cdot\rangle.
\end{multline*}

Using (\ref{ISOex2}), we get $\pi=\Lambda\Pi$.
The mechanical framework to be hereafter studied  goes as follows: the configuration manifold is $SO_3$, whilst $T^{*}SO_3$ is endowed with its canonical symplectic structure. It is remarkable that our framework will retrieve the dynamics of a solid rigid under no exterior forces as a particular autonomous case.

Let us consider a $t$-dependent {\it  Hamiltonian} $h:\mathbb{R}\times T^*SO_3\rightarrow \mathbb{R}$ of the form
\begin{equation}
\label{Hex}
h(t,\Lambda,\widehat{\pi}):=\displaystyle \frac{1}{2}\pi\cdot \mathbb{I}_t^{-1}\pi,\quad  \mathbb{I}_t:=\Lambda \mathbb{J}_t\Lambda^{T}.
\end{equation}
where $\mathbb{I}_t$ is the {\it  time-dependent inertia tensor} (in spatial coordinates) and $\mathbb{J}_t$ is the {\it  inertia dyadic} given by
$
\mathbb{J}_t=\int_{\mathbb{R}^3}\varrho_{\nu}(t,X)[\|X\|^21\!\!1-X\otimes X ]d^{3}X.
$  Here, $\varrho_{\nu}:\mathbb{R}\times\mathcal{B}\rightarrow \mathbb{R}$ is the time-dependent reference density. Note that $\mathbb{J}_t$ can be understood as a matrix depending only on time. Note that $\mathbb{J}_t$ gives, at each $t\in\mathbb{R}$, the natural inertia tensor for our distribution of mass for every time $t\in \mathbb{R}$, is indeed a natural generalisation of its time-independent analogue \cite{MS88}. Our formalism for almost rigid bodies can be applied independently of its explicit form. We understand $h$ in (\ref{Hex}) as a function $h:\mathbb{R}\times SO_3\times\mathfrak{so}^{*}_3\rightarrow \mathbb{R}$, with $\mathfrak{so}^*_3\simeq \mathbb{R}^{3*}$. This is used as $h(t,\Lambda,\widehat\pi)$ is more appropriate for calculations. Note that $h$ is the kinetic energy of the mechanical system, which we call a {\it  quasi-rigid body} (cf. \cite{MS88}). 

Let us study the invariance properties of our Hamiltonian. Since $\widehat{\pi}=\Lambda\widehat{\Pi}\Lambda^T$, the $t$-dependent Hamiltonian  (\ref{Hex}) becomes
\begin{multline}
\label{INVex}
h(t,\Lambda,\widehat\pi)=\frac{1}{4}{\rm tr}(\widehat{\pi}^T\Lambda\mathbb{J}_t^{-1}\Lambda^T\widehat{\pi})=\frac{1}{4}{\rm tr}((\Lambda^T\widehat{\pi})^T\mathbb{J}_t^{-1}\Lambda^T\widehat{\pi})=\\\frac 14{\rm tr}((\widehat{\Pi}\Lambda^T)^T\mathbb{J}_t^{-1}\widehat{\Pi}\Lambda^T)=\frac{1}{4}{\rm tr}(\widehat{\Pi}^T\mathbb{J}_t^{-1}\widehat{\Pi})=\frac{1}{2}\Pi\cdot \mathbb{J}_t^{-1}\Pi,
\end{multline}
which illustrates the {\it  left invariance} of $h$ relative to the action of $SO_3$. Thus, the {\it  left reduction by} $SO_3$ induces a function on the quotient $\mathbb{R}\times T^*SO_3/SO_3\simeq \mathbb{R}\times \mathfrak{so}^*_3$.

As a consequence, $h_t$ is only a quadratic function on the momenta $\widehat{\pi}$. Choosing an appropriate coordinate system adapted to the ${\bf J}^{-1}((\mu)/SO_3)_{\hat \pi}$ and an appropriate $t$-dependent dependence, the second condition in (\ref{eq:Cond}) follows. 

 Momentum map - We consider $G=SO_3$ to act on $Q=SO_3$ by left translations, i.e. $\Psi:(A, \Lambda)\in G\times Q\mapsto L_{A}\Lambda:=A\Lambda\in Q$.
Hence, the {\it  cotangent lift} of $\Psi$, let us say $\widehat{\Psi}$, is by {left translations}. In particular,
\[
\widehat{\Psi}(\Lambda',(\Lambda,\widehat\pi\Lambda))=(\Lambda'\Lambda,\widehat{\Lambda'\pi}\Lambda'\Lambda),\qquad \forall \Lambda',\Lambda\in SO_3,\forall \pi\in (\mathbb{R}^3)^*.
\]
We consider the momentum map associated with our problem as a mapping ${\bf J}:SO_3\times \mathfrak{so}_3^*\rightarrow \mathfrak{so}_3^*$, where we used the identification of $T^*SO_3$ with $SO_3\times \mathfrak{so}_3^*$ via the right-translations $R_\Lambda$, with $\Lambda\in SO_3$. 
Since $(\widehat{\xi}_{\mathfrak{so}_3})_\Lambda=\frac{d}{dt}\big|_{t=0}\exp(t\widehat{\xi})\Lambda=\widehat{\xi}\Lambda$, for every $\xi\in\mathfrak{so}_3$, 
Proposition \ref{LiftSymAc} yields that
\begin{equation}
\label{Jex}
J_{\widehat{\xi}}(\widehat{\pi}{\Lambda})\!=\!\frac{1}{2} {\rm tr} [ ({\Lambda}\widehat{\pi})^{T}\widehat \xi_{\mathfrak{so}_3}]\!=\!\frac{1}{2}{\rm tr}[\Lambda^{T}\widehat{\pi}^{T}\widehat{\xi}\Lambda]\!=\!\frac{1}{2}{\rm tr}[\widehat{\pi}^{T}\widehat{\xi}]\!=\!\pi\cdot\xi.
\end{equation}
Thus, ${\bf J}(\Lambda, \widehat\pi)= \widehat\pi$, $J_{\widehat{\xi}}(\widehat{\pi}\Lambda)=\pi\cdot \xi$. Then, every $\widehat\pi\in \mathfrak{so}_3^*$ is a regular value of ${\bf J}$. Moreover, $G_\pi$ is given by the elements of $SO_3$ that leave invariant $\pi$. Hence, $G_\pi\simeq SO_2$ for $\pi\neq 0$ and $G_0=SO_3$. Moreover ${\bf J}^{-1}(\widehat{\pi})=SO_3\times \{\widehat{\pi}\}$ for every $\widehat{\pi} \in\mathfrak{so}^*_3$. Since each $G_\pi$ is always compact, it acts properly on ${\bf J}^{-1}(\widehat{\pi})$. Moreover, the action of $G_\pi$ on ${\bf J}^{-1}(\widehat{\pi})$ is always free, even for $\hat\pi=0$. Hence, ${\bf J}^{-1}(\widehat{\pi})/G_\pi$ is always a well-defined two-dimensional manifold for $\widehat{\pi}\neq 0$, a sphere indeed, and a zero-dimensional manifold for $\widehat{\pi}=0$.

Let us study
\[
h_{\xi,t}=h_t-[J_{\xi}-\pi_{e}\cdot\xi]=\frac{1}{2}\pi\cdot \mathbb{I}_t^{-1}\pi-\xi\cdot(\pi-\pi_{e}),
\]
and look into its critical points. To derive the first variation,  it is appropriate to consider $h_{\xi,t}$ as a function of $(\Lambda, \pi)\in SO_3\times\mathfrak{so}_3^{*}$. If $(\Lambda_{e},\widehat{\pi}_{e}\Lambda_{e})\in T^*SO_3$ is a relative equilibrium point, then, for any $\delta\theta\in \mathbb{R}^{3}$, we can define the curve 
$\epsilon\mapsto \Lambda_{\epsilon}:=\exp[\epsilon\widehat{\delta\theta}]\Lambda_{e}$ in $ SO_3.$
Let $\widehat{\delta\pi}\in \mathfrak{so}_3^*$ and consider the curve in $\mathfrak{so}_3^{*}$ defined as
$
\epsilon\mapsto \widehat\pi_{\epsilon}:=\widehat\pi_{e}+\epsilon\widehat{\delta\pi}\in \mathfrak{so}_3^{*}.
$
These constructions induce a curve $\epsilon\mapsto(\Lambda_\epsilon,\widehat{\pi}_\epsilon\Lambda_{\epsilon})\in T^{*}SO_3$. Let us compute the first variation.

Let us consider $\delta h_{\xi,t}:=dh_{\xi,t}(\widehat{\delta\theta},\widehat{\delta\pi})$.  
 By using the chain rule and defining $\mathbb{I}_{t,\epsilon}:=\Lambda_{\epsilon}\mathbb{J}_t\Lambda_{\epsilon}^T$, we can establish
\begin{equation}\label{def:FV}
0=\delta h_{\xi,t}\big|_{e}=\frac{d}{d\epsilon}\bigg|_{\epsilon=0}\left(\frac{1}{2}\pi_{\epsilon}\cdot \mathbb{I}_{t,\epsilon}^{-1}\pi_{\epsilon}-\xi\cdot(\pi_{\epsilon}-\pi_e)\right),
\end{equation}
where
$
\mathbb{I}_{t,\epsilon}^{-1}:=\Lambda_{\epsilon}\mathbb{ J}_t^{-1}\Lambda_{\epsilon}^{T}.
$
At equilibrium, considering $h_{\xi,t}$ as a function on $P\times \mathfrak{so}_3$, we obtain  $(\pi-\pi_{e})\cdot\eta=0$ for all $\eta\in \mathbb{R}^{3}$, from varying the Lagrange multiplier.
Recall that
{\small\begin{multline}
\label{compex}
\frac{1}{2}\pi_{e}\cdot\frac{d}{d\epsilon}\bigg|_{\epsilon=0}\mathbb{I}_{t,\epsilon}^{-1}\pi_{e}=\frac{1}{2}\pi_{e}\cdot[\widehat{\delta\theta}\mathbb{I}_{t,e}^{-1}-\mathbb{I}_{t,e}^{-1}\widehat{\delta\theta}]\pi_{e}=\\
\frac{1}{2}[\pi_{e}\cdot(\delta\theta \times\mathbb{I}_{t,e}^{-1}\pi_{e})-\mathbb{I}_{t,e}^{-1}\pi_{e}\cdot(\delta\theta\times\pi_{e})]=\delta\theta \cdot(\mathbb{I}_{t,e}^{-1}\pi_{e}\times\pi_{e}),
\end{multline}}
by using elementary vector product identities. By (\ref{compex}), expression (\ref{def:FV}) reduces to
\begin{equation}\label{deltaHex}
\delta h_{\xi,t}\big|_{e}=\delta\pi\cdot[\mathbb{I}^{-1}_{t,e}\pi_{e}-\xi_t]+\delta\theta\cdot[\mathbb{I}_{t,e}^{-1}\pi_{e}\times\pi_{e}]=0.
\end{equation}
At critical points, the above must vanish for every $\delta\pi$ and $\delta\theta$. This amounts to the conditions
\[
\mathbb{I}^{-1}_{t,e}\pi_{e}=\xi_t,\qquad \mathbb{I}_{t,e}^{-1}\pi_{e}\times\pi_{e}=0.
\]
Hence, substituting the first condition into the second  gives $\xi_t \times \pi_{e}=0$, while the second condition tells us that $\mathbb{I}_{t,e}^{-1}\pi_{e}$ and $\pi_e$ are proportional for every $t$. Hence, $\xi_t$ and $\pi_e$ are also proportional and we can write that $\xi_t=\sigma_t \pi_t$ for a certain $t$-dependent function $\sigma_t$. Hence,
\begin{equation}
\label{eqcondex}
\xi_t \times \pi_{e}=0,\quad\mathbb{I}_{t,e}^{-1}\xi_t=\lambda_t\xi_t,
\end{equation}
where $\lambda_t>0$ due to the positive definiteness of $\mathbb{I}_{t,e}$. These conditions yield that $\pi_{e}$ lays along a principal axis, namely a vector space spanned by an eigenvector of $\mathbb{I}_t$, and that the rotation, recall Hamilton equations, is around this axis. 

Let us study the second variation. By (\ref{deltaHex}), we reach at equilibrium 
\[
(\delta^2 h_{\xi,t})\big|_{e}:=\frac{d}{d\epsilon}\bigg|_{\epsilon=0}[\delta\pi\cdot(\mathbb{I}_{t,\epsilon}^{-1}\pi_{\epsilon}-\xi)+\delta\theta\cdot (\mathbb{I}_{t.\epsilon}^{-1}\pi_{\epsilon}\times\pi_{\epsilon})].
\]
Note that the matrix of second-order derivatives is determined by its value on pairs of equal tangent vectors. 
Proceeding in the same way as to obtain (\ref{deltaHex}) and using (\ref{eqcondex}), we get at equilibrium 
\begin{equation}\label{Eq:SecondVar}
(\delta^{2}h_{\xi,t})\big|_{e}((\delta\pi,\delta \theta),(\delta\pi,\delta\theta))
=
{\small \begin{array}{cc}
 [\delta\pi^{T}\delta\theta^T]\left[\begin{array}{cc}\mathbb{I}_{t,e}^{-1}&(\mathbb{I}_{t,e}^{-1}-\lambda_t 1\!\!1)\widehat{\pi}_e\\-\widehat{\pi}_e(\mathbb{I}^{-1}_{t,e}-\lambda_t 1\!\!1)&-\widehat{\pi}_e(\mathbb{I}_{t,e}^{-1}-\lambda_t 1\!\!1)\widehat{\pi}_e\end{array}\right]&\left[\begin{array}{c}\delta\pi\\\delta\theta\end{array}\right].
\end{array}}
\end{equation}
Let us assume $(\delta\pi, \delta\theta)\in \mathbb{R}^{3*}\times \mathbb{R}^{3}$. We already know that $\textbf{J}(\widehat{\pi}\Lambda)=\widehat{\pi}$. Hence, $\mu_{e}=\widehat{\pi}_{e}$ and $T_{z_e}(G_{\mu_e}z_{e})$ are the generators of infinitesimal rotations around the axis $\pi_{e}$.  Then, one can find different possible $\mathbb{I}_{t,e}$ for which one gets that the application of our results ensure the stability of the reduced problem at the projection of a relative equilibrium point. As an easy example, the $t$-independent case follows exactly as in \cite{MS88}. In particular, the condition on the spatial derivatives of  $h_{\mu_e}$ of third-order or their partial in terms of time are trivially satisfied. More involved examples concern diagonal matrices $\mathbb{I}_{t,e}$ with positive nonincreasing eigenvalues which are properly bounded from below and, in some cases, also from above.  
\section{Conclusions and outlook}

This work has extended the formalism for the energy-momentum method on symplectic manifolds to the non-autonomous realm. This has required the use of $t$-dependent techniques to study the stability of non-autonomous problems. As a byproduct, the formulation of the Lyapunov theory on vector spaces has been extended to manifolds.  Some relations of the energy-momentum method to the theory of foliated Lie systems have been established. A simple example concerning a modification of a rotating quasi-solid rigid has been used to illustrate our techniques.

Note that the energy-momentum method has extensions to look into problems on Poisson manifolds \cite{MS88}. Our techniques should be easily extended to such a new realm. We plan to study the topic in the future. We additionally search for new applications of our techniques in physics.
In particular, we are interested in the study of foliated Lie systems appearing in the study of relative equilibrium points of mechanical systems. Moreover, we are interested in the geometric, i.e. not coordinate dependent, characterisation of conditions for the different types of stability of the projections to $P_{\mu_e}$ of relative equilibrium points.

In the future, we will use our methods to study the motion of a ballet dancer turning around an axis, acrobatic diving into a swimming pool, or other celestial problems like stars passing through a nebula with a variable density. In these cases and many others, we believe that the motion can be effectively described by assuming a $t$-dependent inertia tensor. We expect to study the conditions to be able to do that in a future work. We will focus on the motions of objects that change their shape within some limits that make our formalism appropriate.
\section*{Acknowledgements}
We would like to thank an anonymous referee for his numerous and interesting remarks that undoubtedly helped us to clarify the results of our work. J. de Lucas acknowledges funding from the research project HARMONIA (grant number: 2016/22/M/ST1/00542) financed by the Polish National Science Centre (POLAND).

\end{document}